\documentclass[printer,useAMS]{gLMA2ePublic}

\usepackage{xspace} 
\usepackage{xcolor}

\def\herm{\text{$\mathrm{H}$}}
\def\inv{\text{$-1$}}

\newcommand{\mps}{{\sc{mps}}\xspace}
\newcommand{\MPS}{{\sc{mps}}\xspace}

\newcommand{\TIMPS}{{\sc{ti mps}}\xspace}

\newcommand{\DMRG}{{\sc{dmrg}}\xspace}

\def\defini#1{\emph{#1}}

\renewcommand{\vec}[1]{\bm{#1}}
\newcommand{\mat}[1]{{\bf{#1}}}

\newcounter{subequation}
\newlength\mtabskip\mtabskip=-1.25cm

\def\mtabLong{long}


\newcommand{\trace}{\operatorname{tr}}

\newcommand{\T}{\text{$\mathrm{T}$}}

\newcommand{\diag}{\operatorname{diag}}

\newcommand{\myref}[1]{(\ref{#1})}

\providecommand{\pmat}[1]{\begin{pmatrix} #1 \end{pmatrix}}

\begin{document}

\markboth{Huckle, Waldherr, Schulte-Herbr{\"u}ggen}{Exploiting Matrix Symmetries and Physical Symmetries in Matrix Product States}


\title{Exploiting Matrix Symmetries and Physical Symmetries in Matrix Product States and Tensor Trains}

\author{
Thomas K. Huckle$^{\rm a}$ and
Konrad Waldherr$^{\rm a}$$^{\ast}$\thanks{$^\ast$Corresponding author. Email: waldherr@in.tum.de \vspace{6pt}}
and Thomas Schulte-Herbr{\"u}ggen$^{\rm b}$ \\\vspace{6pt}
$^{\rm a}${Technische Universit\"at M\"unchen, Boltzmannstr. 3, 85748 Garching, Germany};
$^{\rm b}${Technische Universit\"at M\"unchen, Lichtenbergstr. 4, 85748 Garching, Germany}
}

\maketitle

\begin{abstract}
We focus on symmetries related to matrices and vectors appearing
in the simulation of quantum many-body systems.
Spin Hamiltonians have special matrix-symmetry properties such as persymmetry.
Furthermore, the systems may exhibit physical symmetries translating into
symmetry properties of the eigenvectors of interest.
Both types of symmetry can be exploited in sparse representation formats such as
Matrix Product States (\MPS) for the desired eigenvectors.

This paper summarizes symmetries of Hamiltonians for typical
physical systems such as the Ising model and lists resulting properties of the related eigenvectors.
Based on an overview of Matrix Product States (Tensor Trains or Tensor Chains) and their canonical normal forms
we show how symmetry properties of the vector
translate into relations between the \MPS matrices and, in turn,
which symmetry properties result from relations within the MPS matrices.
In this context we analyze different kinds of symmetries
and derive appropriate normal forms for MPS representing these symmetries.
Exploiting such symmetries by using these normal forms will lead to a reduction
in the number of degrees of freedom in the MPS matrices.
This paper provides a uniform platform for both well-known and new results
which are presented from the (multi-)linear algebra point of view.

\begin{keywords}
Symmetric persymmetric matrices;
Quantum many-body systems;
Spin Hamiltonian;
Matrix Product States;
Tensor Trains;
Tensor Chains
\end{keywords}
\begin{classcode}
15A69
15B57;
81-08;
15A18
\end{classcode}\bigskip

\end{abstract}

\section{Introduction}
\label{sec:intro}
In the simulation of quantum many-body systems such as 1D spin chains
one is faced with problems growing exponentially in the system size.
From a linear algebra point of view, the physical system can be described by a Hermitian matrix $\mat H$,
the so-called \defini{Hamiltonian}.
The real eigenvalues of $\mat H$ correspond to the possible energy levels of the system,
the related eigenvectors describe the corresponding states.
The \defini{ground state} is of important relevance because it is related to the state of minimal energy
which naturally arises.
To overcome the exponential growth of the state space with system size (sometimes referred to as
\/`curse of dimensionality\/')
one uses sparse representation formats that scale only polynomially in the number of particles.
In quantum physics concepts like Matrix Product States have been developed, see, e.g., \cite{PerezGarcia07MPS}.
These concepts strongly relate to the Tensor-Train concept,
which was introduced by Oseledets in \cite{Oseledets11tt}
as an alternative to the canonical decomposition \cite{CarrollChang70Analysis, Harshman70Foundations}
and the Tucker format \cite{Tucker66MathematicalNotes}.

In the \MPS formalism vector components are represented by the trace of a product of matrices,
which are often of moderate size.
As will turn out,
symmetries and further relations in these matrices result in special properties of the vectors
to be represented and, vice versa, that special symmetry properties of vectors can be
expressed by certain relations of the \MPS matrices.
We will analyze different symmetries such as the bit-shift symmetry, the reverse symmetry,
and the bit-flip symmetry,
and we present normal forms of \MPS for these symmetries,
which will lead to a reduction of the degrees of freedom in the decomposition schemes.

\subsubsection*{Organization of the Paper}
The paper is organized as follows:
First, we list pertinent matrix symmetries
translating into symmetry properties of their eigenvectors.
Then we consider physical model systems and summarize the related symmetries
translating into symmetries of the eigenvectors.
After a fixing notation of Matrix Product States,
we present normal forms of \MPS and analyze how relations
between the \MPS matrices and symmetries of the represented vectors are interconnected.
Finally, we show the amount of data reduction by exploiting symmetry-adapted normal forms.

\section{Matrix Symmetries}
\label{sec:structuredMatrices}
In this section we recall some classes of structured matrices and list some important properties.
A matrix~$\mat A$ is called \defini{symmetric}, if $\mat A = \mat A^{\mathrm T}$ (i.e. $a_{i,j} = a_{j,i}$)
and \defini{skew-symmetric}, if $\mat A^{\mathrm T} = -\mat A$.
A real-valued symmetric matrix has real eigenvalues and a set of orthogonal eigenvectors.
If $\mat{A}$ is symmetric about the ``northeast-to-southwest'' diagonal,
i.e. $a_{i,j} = a_{n-j+1,n-i+1}$,
it is called \defini{persymmetric}.
Let $\mat J \in \mathbb R^{n \times n}, \mat {J}_{i,j} := \delta_{i,n+1-j}$,
be the \defini{exchange matrix}.
Then persymmetry can also be expressed by
\begin{equation*}
\mat{J A J} = \mat{A}^{\T} \; .
\end{equation*}
A matrix is \defini{symmetric persymmetric}, if it is symmetric about both diagonals,
i.e.
$$ \mat{J A J} = \mat{A}^{\mathrm T} = \mat{A}$$
or component-wise
$$ a_{i,j} = a_{j,i} = a_{n+1-i,n+1-j} \; .$$
Note that a matrix with the property $\mat{JAJ}=\mat{A}$ is called \defini{centrosymmetric}.
Therefore, symmetric persymmetric or symmetric centrosymmetric are the same.

The set of all symmetric persymmetric $n \times n$ matrices is closed under addition and under scalar multiplication.

A matrix $\mat A$ is called \defini{symmetric skew-persymmetric} if $\mat{JAJ}=-\mat{A}^{\mathrm T}=-\mat A$,
or component-wise
$$ a_{i,j} = a_{j,i} = -a_{n+1-i,n+1-j} \; .$$
The set of these matrices is again closed under addition and scalar multiplication.

Any symmetric $n \times n$ matrix $\mat A$ can be expressed as a sum of a persymmetric and a skew-persymmetric matrix:
\begin{equation*}
\mat {A} = \tfrac{1}{2} \left( \mat A + \mat{J A J} \right) + \tfrac{1}{2} \left( \mat{A} - \mat{J A J} \right) \; .
\end{equation*}

By $\mat{J}$ one may likewise characterize vector symmetries: a vector $\vec v \in \mathbb R^{n}$
is \defini{symmetric} if $\mat{J} \vec{v} = \vec{v}$ and \defini{skew-symmetric}
if $\mat{J} \vec{v} = - \vec{v}$.

As all the matrices of subsequent interest are built by linear combinations of Kronecker products of smaller matrices
the following lemma will be useful.
\begin{lemma}\label{lemma:KroneckerPersymm}
The Kronecker product of two symmetric persymmetric matrices $\mat{B}$ and $\mat{C}$ is again symmetric persymmetric.
\end{lemma}
\begin{proof}
Let $\mat{J_B}$ and $\mat{J_C}$ denote the exchange matrices
which correspond to the size of $\mat{B}$ and $\mat{C}$ respectively.
Then the exchange matrix $\mat{J}$ of $\mat{B} \otimes \mat{C}$
is given by $\mat{J}=\mat{J_B}\otimes \mat{J_C}$.
Therefore
$$ (\mat{J_B} \otimes \mat{J_C}) (\mat{B} \otimes \mat{C}) (\mat{J_B} \otimes \mat{J_C}) = (\mat{J_B B J_B})
\otimes (\mat{J_C C J_C}) = (\mat{B}^{\T} \otimes \mat{C}^{T}) = \mat{B} \otimes \mat{C} \; .$$
\end{proof}

\begin{remark}\label{rem:powerspersymm}
Each power $\mat A^k$ of a symmetric persymmetric $\mat A$ is again symmetric
persymmetric.
\end{remark}

\begin{remark}\label{rem:skewpersymm}
For a symmetric skew-persymmetric  $\mat A$,
$\mat A^2$ is symmetric persymmetric, and also the Kronecker product of
two symmetric skew-persymmetric matrices is symmetric persymmetric.
\end{remark}

\begin{remark}\label{rem:skewsymmetric}
If matrix $\mat A$ is skew-symmetric, then $\mat{A}^2$ is symmetric.
Furthermore, the Kronecker product of two skew-symmetric matrices is symmetric.
\end{remark}

Due to \cite{CantoniButler} we can state various properties
for symmetric persymmetric matrices and the related eigenvectors.
As all the matrices of our interest have as size a power of $2$,
we focus on the statements related to even matrix sizes here.
The following lemma points out the main results adapted from \cite{CantoniButler}.
Both the proof and similar results for the odd case can be found in the original paper.
\begin{lemma}[(\cite{CantoniButler}) ]\label{lemma:EigenvectorsPersymmMatrices}
Let $\mat A \in \mathbb R^{n \times n}$ be any symmetric persymmetric matrix of even size $n=2m$,
the following properties hold.
\begin{itemize}
\item[a)] The matrix $\mat{A}$ can be written as
$$ \mat A = \left(
     \begin{array}{cc}
       \mat B \ & \ \mat C^{\T} \\
       \mat{C} \ & \ \mat{J B J} \\
     \end{array}
   \right)
$$
with block matrices $\mat B$ and $\mat C$ of size $m \times m$, where $\mat{B}$ is symmetric and $\mat C$ is persymmetric,
i.e. $\mat C^{\T} = \mat{J C J}$.
\item[b)] The matrix $\mat A$ can be orthogonally transformed to a block diagonal matrix
with blocks of half size $m$:
\begin{align} \nonumber
& \ \frac{1}{2}
\left(
  \begin{array}{cc}
    \mat I \ & \ \mat J \\
    \mat I \ & \ - \mat J \\
  \end{array}
\right) \left(
          \begin{array}{cc}
            \mat B \ & \ \mat C^T \\
            \mat C \ & \ \mat{J B J} \\
          \end{array}
        \right)
        \left(
  \begin{array}{cc}
    \mat I \ & \ \mat I \\
    \mat J \ & \ -\mat J \\
  \end{array}
\right) \\ \nonumber
            = & \ \frac{1}{2}
    \left(
      \begin{array}{cc}
        \mat B + \mat{J C} + \mat C^{\mathrm T} \mat J  + \mat B \ & \ \mat B + \mat{J C} - \mat{C}^{\T} \mat J - \mat B \\
        \mat B - \mat{J C} + \mat{C}^{\T} \mat J - \mat B \ &\ \mat B - \mat{J C} - \mat{C}^{\T} \mat J + \mat B \\
      \end{array}
    \right) \\ \label{eq:Transform2BlocksEvenN}
    = & \ \left(
          \begin{array}{cc}
            \mat B + \mat{J C} \ & \ \mat 0 \\
            \mat 0 \ & \ \mat B - \mat{J C} \\
          \end{array}
        \right) \; .
\end{align}
\item[c)] The matrix $\mat A$ has $m$ skew-symmetric orthonormal eigenvectors of the form
$$ \frac{1}{\sqrt{2}} \left(
     \begin{array}{c}
       \vec{u_i} \\
       - \mat J \vec{u_i} \\
     \end{array}
   \right) \; ,
$$
where $\vec{u_i}$ are the orthonormal eigenvectors of $\mat B - \mat{J C}$.\\
$\mat A$ also has $m$ symmetric orthonormal eigenvectors
$$ \frac{1}{\sqrt{2}} \left(
     \begin{array}{c}
       \vec{v_i} \\
       \mat J \vec{v_i} \\
     \end{array}
   \right) \; ,
$$
where the $\vec{v_i}$ are the orthonormal eigenvectors of $\mat B + \mat{J C}$.
\end{itemize}

\end{lemma}
The discussed transformation (\ref{eq:Transform2BlocksEvenN})
to block diagonal matrices of smaller size is quite cheap
and can be exploited to save computational costs, see, e.g., \cite{Auckenthaler08ExpmPrefix}.
\begin{remark}
In general, the transformation of symmetric persymmetric matrices to block diagonal form
\myref{eq:Transform2BlocksEvenN} cannot be continued recursively
because the matrix $\mat B \pm \mat{JC}$ is symmetric but usually no longer persymmetric.
\end{remark}

Altogether, any symmetric persymmetric matrix has eigenvectors which are either symmetric or skew-symmetric,
i.e. $\mat J \vec{v} = \vec{v}$ or $\mat J \vec{v} = - \vec{v}$.
However, one has to be careful with these statements in the case of degenerate 
eigenvalues.
If the two blocks share an eigenvalue, $\mat A$ has as eigenvectors linear combinations of symmetric
and skew-symmetric vectors, so the eigenvectors themselves are in general neither symmetric nor skew-symmetric.

A matrix is called \defini{Toeplitz matrix}, if it is of the form
\begin{equation*}\label{eq:Toeplitz}
\mat T = \left(
      \begin{array}{cccc}
        r_0 & r_1 &  & r_{n-1} \\
        r_{-1} & r_0 & \ddots &  \\
         & \ddots & \ddots & r_1 \\
        r_{-n+1} &  & r_{-1} & r_0 \\
      \end{array}
    \right) \; .
\end{equation*}
Toeplitz matrices obviously belong to the larger class of persymmetric matrices.
Therefore, real symmetric Toeplitz matrices are symmetric persymmetric.
An important class of Toeplitz matrices are the \defini{circulant} matrices
taking the form
\begin{equation*}\label{eq:circulant}
\mat C = \left(
      \begin{array}{cccc}
        r_0 & r_1 &  & r_{n-1} \\
        r_{n-1} & r_0 & \ddots &  \\
         & \ddots & \ddots & r_1 \\
        r_{1} &  & r_{n-1} & r_0 \\
      \end{array}
    \right) \; .
\end{equation*}
Any circulant matrix $\mat C$ with entries $\vec r := (r_0, r_1, \dots, r_{n-1})^{\mathrm T}$ can be diagonalized
by the Fourier matrix $\mat{F_n} = (f_{j,k}); f_{j,k} = \tfrac{1}{\sqrt{n}} e^{2 \pi i j k / n}$ \cite{GolubLoan} via
\begin{equation}\label{eq:DiagonalizeCirculantDFT}
\mat C = \mat {F_n^{\text{$-1$}}} \diag( \mat{F_n} \vec{r}) \mat{F_n} = \mat{F_n} \diag( \mat{F_n} \vec{r}) \mat{F_n} \; .
\end{equation}
Analogously, a \defini{skew-circulant} matrix looks like
\begin{equation*}\label{eq:skewCirculant}
\mat{C_s} = \left(
      \begin{array}{cccc}
        r_0 & r_1 &  & r_{n-1} \\
        -r_{n-1} & r_0 & \ddots &  \\
         & \ddots & \ddots & r_1 \\
        -r_{1} &  & -r_{n-1} & r_0 \\
      \end{array}
    \right) \; .
\end{equation*}
In general, an \defini{$\omega$-circulant} matrix with $\omega = e^{i \phi}$ is defined by
\begin{equation*}\label{eq:omCirculant}
\mat{C_{\boldsymbol \omega}} = \left(
      \begin{array}{cccc}
        r_0 & r_1 &  & r_{n-1} \\
        \omega r_{n-1} & r_0 & \ddots &  \\
         & \ddots & \ddots & r_1 \\
        \omega r_{1} &  & \omega r_{n-1} & r_0 \\
      \end{array}
    \right) \; .
\end{equation*}
These matrices can be transformed into a circulant matrix by the unitary
diagonal matrix $\mat{\Omega_{n; \boldsymbol \omega}} = \diag( \omega^{j/n} )_{j=0,...,n-1}$ :
\begin{equation}\label{eq:transformomegaCirculant2Circulant}
\mat{\Omega_{n; \boldsymbol \omega}^{\herm}} \mat{C_{\boldsymbol \omega}} \mat{\Omega_{n; \boldsymbol \omega}}
    = {\mat{\overline{\Omega}}_{n; \boldsymbol \omega}} \mat{C_{\boldsymbol \omega}} \mat{ \Omega_{n; \boldsymbol \omega}}
    = \left(
      \begin{array}{cccc}
        \tilde r_0 & \tilde r_1 &  & \tilde r_{n-1} \\
        \tilde r_{n-1} & \tilde r_0 & \ddots &  \\
         & \ddots & \ddots & \tilde r_1 \\
        \tilde r_{1} &  & \tilde r_{n-1} & \tilde r_0 \\
      \end{array}
    \right) \; ,
\end{equation}
where $\tilde r_k := \omega^{k/n} r_k$.
\defini{Multilevel circulant} matrices are defined by the property that
the eigenvector matrix is given by a tensor product of Fourier matrices
\mbox{$\mat{F_{n_1}} \otimes \cdots \otimes \mat{F_{n_k}}$}.
\defini{Block-Toeplitz-Toeplitz-Block} matrices, also called \defini{2-level Toeplitz} matrices,
have a Toeplitz block structure where each block itself is Toeplitz.
More general, a \defini{multilevel Toeplitz matrix} has a hierarchy of blocks
with Toeplitz structure.

\subsection{Representations of Spin Hamiltonians}
\label{subsec:HamiltonianRepresentations}
For spin-$\tfrac{1}{2}$ particles such as electrons or protons, the spin angular momentum operator
describing their internal degree of freedom (i.e. \defini{spin-up} and \defini{spin-down})
is usually expressed in terms of the \defini{Pauli matrices}
\begin{equation*}\label{eq:Pauli}
\mat{P_x}= \pmat{0 \ & 1 \\ 1 \ & 0 }, \
\mat{P_y}= \pmat{0 \ & -i \\ i \ & 0 } = i \pmat{0 \ & -1 \\ 1 \ & 0} \ \text{and} \
\mat{P_z}= \pmat{1 \ & 0 \\ 0 \ & -1}\quad.
\end{equation*}
For further details, a reader wishing to
approach quantum physics from linear and multilinear algebra may refer to
\cite{Faddeev2009Lectures}. 
Being traceless and Hermitian, $\{\mat{P_x}, \mat{P_y},\mat{P_z}\}$ forms a basis of the Lie algebra $\mathfrak{su}(2)$,
while by appending the $2\times2$ identity matrix $\mat I$ one obtains a basis of the Lie algebra $\mathfrak{u}(2)$.
This fact can be generalized in the following way:
for any integer $p$ a basis of the Lie algebra $\mathfrak{u}(2^p)$ is given by
$$ \bigl\{ \mat{Q_1} \otimes \mat{Q_2} \otimes \cdots \otimes \mat{Q_p} \; ; \;
\mat{Q_i} \in \{\mat{P_x}, \mat{P_y}, \mat{P_z},\mat I\} \bigr\} \; .$$
To get a basis for $\mathfrak{su}(2^p)$ we have to consider only traceless matrices
and therefore we have to exclude the identity, which results in the basis
$$ \bigl\{ \mat{Q_1} \otimes \mat{Q_2} \otimes \cdots \otimes \mat{Q_p} \; ; \;
\mat{Q_i} \in \{\mat{P_x}, \mat{P_y}, \mat{P_z},\mat{I}\} \bigr\}
    \setminus \{ \mat{I} \otimes \cdots \otimes \mat{I} \} \; . $$

Now, spin Hamiltonians are built by summing $M$ terms,
each of them representing a physical (inter)action.
These terms are themselves tensor products of Pauli matrices or identities
\begin{equation}\label{eq:Hamiltonian}
\mat H = \sum_{k=1}^M\underbrace{\alpha_k \bigl(\mat{Q_{1}^{(k)}}\otimes \mat{Q_{2}^{(k)}}\otimes %
		\cdots \otimes \mat{Q_{p}^{(k)}} \bigr)}_{=:\mat{H^{(k)}}}= \sum_{k=1}^M \mat{H^{(k)}}\;,
\end{equation}
where the coefficients $\alpha_k$ are real and the matrices
$\mat{Q_j^{(k)}}$ can be $\mat{P_x}$, $\mat{P_y}$, $\mat{P_z}$ or $\mat I$.

In each summand $\mat{H^{(k)}}$ most of the $\mat{Q_{j}^{(k)}}$ are $\mat I$:
{\em local terms} have just one nontrivial tensor factor, while {\em pair interactions} have two of them.
Higher \mbox{$m$-body} interactions (with $m>2$) usually do not occur as physical primitives, but could be
represented likewise by $m$ Pauli matrices in the tensor product representing the $m$-order interaction
term.
For defining spin Hamiltonians we will need tensor powers of the $2 \times 2$ identity~$\mat I$:
$$ \mat{I}^{\otimes k} := \underbrace{\mat I \otimes \cdots \otimes \mat I}_{k} \; .$$
For instance, in the \defini{Ising} ($ZZ$) model (see e.g. \cite{P70}) for the 1D chain with $p$ spins
and open boundary conditions, the spin Hamiltonian takes the form
\begin{equation}\label{eq:IsingModel}
\begin{split}
\mat H  & = \sum_{k=1}^{p-1} \mat I^{\otimes (k-1)}\otimes
        (\mat{P_z})_k\otimes (\mat{P_z})_{k+1}\otimes \mat{I}^{\otimes (p-k-1)} \\
   & \qquad +  \lambda \sum_{k=1}^p \mat{I}^{\otimes (k-1)}\otimes (\mat{P_x})_k\otimes \mat{I}^{\otimes (p-k)}\; ,
\end{split}
\end{equation}
where the index $k$ denotes the position in the spin chain and the real number $\lambda$ describes
the ratio of the strengths of the magnetic field and the pair interactions.
Using $\mu,\nu\in\{x,y,z\}$, one may define
{\allowdisplaybreaks\begin{align}
\mat{H_\nu} & := \sum_{k=1}^{p} \mat{I}^{\otimes (k-1)} \otimes
                    (\mat{P_\nu})_k \otimes \mat I^{\otimes (p-k) } \; , \label{eq:defHx} \\
 \mat{H_{\mu\mu}}& := \sum_{k=1}^{p-1} \mat I^{\otimes (k-1)} \otimes
                    (\mat{P_\mu})_k \otimes (\mat{P_\mu})_{k+1} \otimes \mat I^{\otimes (p-k-1)} \; . \label{eq:defHxx}
\end{align}}
The terms \myref{eq:defHxx} correspond to the so-called \defini{open boundary case}.
In the \defini{periodic boundary case} there are also connections between sites $1$ and $p$, which reads
\begin{align}\label{eq:defHxxPBC}
\mat{H_{\mu\mu}^{'}} = \mat{H_{\mu\mu}} + (\mat{P_\mu})_1 \otimes \mat I^{\otimes (p-2)} \otimes (\mat{P_\mu})_p \; .
\end{align}
Note that in the literature often the identity matrices and the tensor
products are ignored giving the equivalent notation
\begin{align}\label{eq:defHxxPBCI}
\mat{H_{\mu\mu}^{'}} & := \sum_{k=1}^{p}(\mat{P_\mu})_k  (\mat{P_\mu})_{k+1 \ \mathrm{mod} \ p}   \; .
\end{align}
In analogy to the Ising model \myref{eq:IsingModel}, it is customary to define
various types of Heisenberg models (\cite{LSM61,AKLT87})
in terms of (either vanishing or degenerate) real constants $j_x$, $j_y$ and $j_z$.
Table \ref{tab:Heisenberg} gives a list of possible 1D models
\begin{table}
  \tbl{List of different 1D models.}
{\begin{tabular}{@{}ll}\toprule
   Interaction  & Hamiltonian \\
\colrule
  Ising-ZZ & $j_z \mat{H_{zz}} + \lambda \mat{H_x}$ \\[2mm]
  Heisenberg-XX & $j_x \mat{H_{xx}} + j_x \mat{H_{yy}} + \lambda \mat{H_x}$ \\[1mm]
  Heisenberg-XY & $j_x \mat{H_{xx}} + j_y \mat{H_{yy}} + \lambda \mat{H_x}$ \\[1mm]
  Heisenberg-XZ & $j_x \mat{H_{xx}} + j_z \mat{H_{zz}} + \lambda \mat{H_x}$ \\[1mm]
  Heisenberg-XXX & $j_x \mat{H_{xx}} + j_x \mat{H_{yy}} + j_x \mat{H_{zz}} + \lambda \mat{H_x}$ \\[1mm]
  Heisenberg-XXZ & $j_x \mat{H_{xx}} + j_x \mat{H_{yy}} + j_z \mat{H_{zz}} + \lambda \mat{H_x}$ \\[1mm]
  Heisenberg-XYZ & $j_x \mat{H_{xx}} + j_y \mat{H_{yy}} + j_z \mat{H_{zz}} + \lambda \mat{H_x}$ \\[1mm]
   \botrule
  \end{tabular}}
\label{tab:Heisenberg}
\end{table}
where, in addition, one may have either open or periodic boundary conditions.
The operators with the additional term $\lambda \mat{H_x}$ are sometimes called
generalized  Heisenberg models. The $XX$, resp. $XXX$ models are called
\defini{isotropic}.

For spin-$1$ models, the operators take the form
\begin{equation}\label{eq:spin1Operators}
\mat{S_x} = \frac{1}{\sqrt{2}} \pmat{0 \ & \ 1 \ & \ 0 \\ 1 \ & \ 0 \ & \ 1 \\ 0 \ & \ 1 \ & \ 0} \; , \;
\mat{S_y} = \frac{1}{\sqrt{2}} \pmat{0 \ & \ -i \ & \ 0 \\ i \ & \ 0 \ & \ -i \\ 0 \ & \ i \ & \ 0} \; , \;
\mat{S_z} = \pmat{1 \ & \ 0 \ & \ 0 \\ 0 \ & \ 0 \ & \ 0 \\ 0 \ & \ 0 \ & \ -1} \; .
\end{equation}

The AKLT model is defined as (\cite{LSM61,AKLT87})
\begin{equation}\label{eq:defAKLT}
\mat{H}=\sum_k \vec S_k\vec S_{k+1}+\tfrac{1}{3}(\vec S_k\vec S_{k+1})^2
\end{equation}
where
$\vec{S_k} \vec{S_{k+1}} := (\mat{S_x})_k (\mat{S_x})_{k+1} +
(\mat{S_y})_k (\mat{S_y})_{k+1} + (\mat{S_z})_k (\mat{S_z})_{k+1}$.
More generally, the bilinear biquadratic model has Hamiltonian
\begin{equation}\label{eq:defBilinearBiquadratic}
\mat{H}=\sum_k \cos(\theta )\vec S_k\vec S_{k+1}+\sin(\theta )(\vec S_k\vec S_{k+1})^2\; .
\end{equation}
These 1D models can also be extended to 2 and higher dimensions.
Then the neighbor relations cannot be represented linearly but they appear in each direction.
For example, Eqn. \ref{eq:defHxx} would read
$$\mat{H_{\mu\mu}} = \sum_{<j,k>} (\mat{P_\mu})_j (\mat{P_\mu})_k \; ,$$
where $<j,k>$ denotes an interaction between particles $j$ and $k$.

Being a sum \myref{eq:Hamiltonian} of Kronecker products of structured $2 \times 2$ matrices,
many Hamiltonians have special properties, e.g.,
they can be multilevel-circulant (\cite{Davis94Circulant,Tyrtyshnikov00Circulant}) or skew-circulant,
diagonal or persymmetric (\cite{CantoniButler}),
which can be exploited to derive properties of the respective eigenvalues and eigenvectors.

\subsection{Symmetry Properties of the Hamiltonians}
\noindent
To begin, we list some properties of the Pauli matrices.
\subsubsection*{Properties of the Pauli Matrices}
$\mat{P_x}$ is symmetric persymmetric and circulant.
Following Eqn. \ref{eq:DiagonalizeCirculantDFT}, $\mat{P_x}$ can be diagonalized via the Fourier matrix $\mat{F_2}$:
\begin{equation}\label{eq:transformPx}
 \mat{F_2 P_x F_2} = \frac{1}{2} \left(
                                    \begin{array}{cc}
                                      1 \ & \ 1 \\
                                      1 \ & \ -1 \\
                                    \end{array}
                                  \right) \left(
                                            \begin{array}{cc}
                                              0 \ & \ 1 \\
                                              1 \ & \ 0 \\
                                            \end{array}
                                          \right) \left(
                                    \begin{array}{cc}
                                      1 \ & \ 1 \\
                                      1 \ & \ -1 \\
                                    \end{array}
                                  \right) = \left(
                                              \begin{array}{cc}
                                                1 \ & \ 0 \\
                                                0 \ & \ -1 \\
                                              \end{array}
                                            \right) = \mat{P_z} \; .
\end{equation}
The matrix $\mat{P_y}/i$ is skew-symmetric persymmetric.
$\mat{P_y}$ is skew-circulant and
by using Eqn. \ref{eq:transformomegaCirculant2Circulant},
it can be transformed into a circulant (and even real) matrix:
\begin{equation}\label{eq:transformPy2Px}
\mat{\overline{\Omega}_{2;-1} P_y \Omega_{2;-1}} = \left(
     \begin{array}{cc}
       1 \ & \ 0 \\
       0 \ & \ -i \\
     \end{array}
   \right) \left(
             \begin{array}{cc}
               0 \ & \ -i \\
               i \ & \ 0 \\
             \end{array}
           \right) \left(
                     \begin{array}{cc}
                       1 \ & \ 0 \\
                       0 \ & \ i \\
                     \end{array}
                   \right) = \left(
                               \begin{array}{cc}
                                 0 \ & \  1 \\
                                 1 \ & \ 0 \\
                               \end{array}
                             \right) = \mat{P_x} \; ,
\end{equation}
which is due to \myref{eq:transformPx} orthogonally similar to $\mat{P_z}$.

$\mat{P_z}$ is diagonal and symmetric skew-persymmetric.
The $2 \times 2$ identity matrix~$\mat{I}$ is of course circulant, symmetric persymmetric and diagonal.

Now we list symmetry properties of the matrices
defined in Eqn. (\ref{eq:defHx}) and (\ref{eq:defHxxPBC}).
As the matrices are built by Kronecker products of $2 \times 2$-matrices
it will be useful to exploit the fact that the exchange matrix can also be expressed
as Kronecker product of $2 \times 2$-matrices:
$$ \mat{J_{2^p}} = \mat{J_2}\otimes \cdots \otimes \mat{J_2} = \mat{P_x} \otimes \cdots \otimes \mat{P_x} \; .$$
Due to Lemma~\ref{lemma:KroneckerPersymm} applied on this factorization
the matrix $\mat{H_x}$ --- as a sum of Kronecker products of symmetric persymmetric matrices ---
is again symmetric persymmetric.
Moreover, $\mat{H_x}$ is multilevel-circulant as it can be diagonalized
by the Kronecker product of the $2\times2$ Fourier matrix $\mat{F_2}$:
\begin{equation}\label{eq:transformHx2Hz}
\begin{split}
& \left( \mat{F_2} \otimes \cdots \otimes  \mat{F_2}  \right)
    \left( \sum_{k=1}^{p} \mat{I}^{\otimes (k-1)} \otimes (\mat{P_x})_k \otimes \mat I^{\otimes (p-k)}\right)
    (\mat{F_2} \otimes \cdots \otimes \mat{F_2}) \\
= \ & \sum_{k=1}^{p} (\underbrace{\mat{F_2 I F_2}}_{=  \mat{I}})^{\otimes (k-1)} \otimes
\underbrace{(\mat{F_2 P_x F_2})_k}_{\stackrel{\myref{eq:transformPx}}{=} (\mat{P_z})_k}
\otimes ( \underbrace{\mat{F_2 I F_2}}_{=\mat I})^{\otimes (p-k)}  \\
= \ & \sum_{k=1}^{p} \mat I^{\otimes (k-1) } \otimes (\mat{P_z})_k  \otimes \mat{I}^{\otimes (p-k)}
        \stackrel{\myref{eq:defHx}}{=} \mat{H_z} \; .
\end{split}
\end{equation}
Therefore the eigenvalues of $\mat{H_x}$ are all $2^p$ possible combinations
$$ \pm 1 \pm 1 \pm \cdots \pm 1 \; .$$
Trivially, the matrix $\mat{H_y}/i$ is skew-symmetric persymmetric and thus $\mat{H_y}$ is Hermitian.
It can be transformed to $\mat{H_x}$ via
the Kronecker product of the diagonal transforms considered in Eqn. \myref{eq:transformPy2Px}.

Even for the generalized anisotropic case $\mat{H^{an}} = \mat{H_x^{{an}}} + \mat{H_y^{{an}}}$,
where each summand $k$ in both sums 
may have an individual coefficient
$a_k$ and $b_k$, respectively, one can find an appropriate transform.
To this end, consider
{\allowdisplaybreaks
\begin{align*}
 \mat{H^{an}}& = \sum_{k=1}^{p} a_k \cdot \mat{I}^{\otimes (k-1)} \otimes (\mat{P_x})_k \otimes \mat{I}^{\otimes (p-k)}
                    + \sum_{k=1}^{p} b_k \mat I^{\otimes (k-1)} \otimes (\mat{P_y})_k \otimes \mat{I}^{\otimes (p-k)} \\
& = \sum_{k=1}^{p} \mat I^{\otimes (k-1)} \otimes \bigl( a_k (\mat{P_x})_k + b_k (\mat{P_y})_k \bigr)
                        \otimes \mat{I}^{\otimes (p-k)} \\
& = \sum_{k=1}^{p} \mat{I}^{\otimes (k-1)} \otimes \left(
                               \begin{array}{cc}
                                 0 & a_k - i b_k \\
                                 a_k + i b_k & 0 \\
                               \end{array}
                             \right)
 \otimes \mat{I}^{\otimes (p-k)} \\
 & = \sum_{k=1}^{p} \mat{I}^{\otimes (k-1)} \otimes \left(
                               \begin{array}{cc}
                                 0 & r_k e^{- i \phi_k} \\
                                 r_k e^{i \phi_k} & 0 \\
                               \end{array}
                             \right)
 \otimes \mat{I}^{\otimes (p-k)} \; .
\end{align*}}
Each tensor factor
$$ \mat{C_k} := \left(
                               \begin{array}{cc}
                                 0 & r_k e^{- i \phi_k} \\
                                 r_k e^{i \phi_k} & 0 \\
                               \end{array}
                             \right) = \left(
                                         \begin{array}{cc}
                                           0 & r_k e^{- i \phi_k} \\
                                           e^{2 i \phi_k}  (r_k e^{- i \phi_k}) & 0 \\
                                         \end{array}
                                       \right)
                             $$
is $\omega$-circulant ($\omega_k = e^{2 i \phi_k}$).
Following \myref{eq:transformomegaCirculant2Circulant},
$\mat{C_k}$ can be transformed to a real matrix using the diagonal transform
$\mat{D_k} = \mat{\Omega_{2;\boldsymbol \omega_k}}$:
\begin{equation*}
 \mat{\bar D_k} \mat{C_k D_k} = \left(
          \begin{array}{cc}
            0 \ & \ r_k \\
            r_k \ & \ 0 \\
          \end{array}
        \right) = r_k \mat{P_x} \; .
\end{equation*}
Therefore, the overall Hamiltonian $\mat{H_x^{an}} + \mat{H_y^{an}}$ can be transformed to an anisotropic $\mat{H_x}$ term:
\begin{eqnarray*}
\lefteqn{\left( \mat{\bar D_1} \otimes \cdots \otimes \mat{\bar D_p} \right)
    \left(  \sum_{k=1}^{p} \mat{I}^{\otimes (k-1)} \otimes \mat{C_k}
            \otimes \mat{I}^{\otimes (p-k)} \right)
    \left( \mat{D_1} \otimes \cdots \otimes \mat{D_p} \right)} \\
& = & \sum_{k=1}^{p} \mat{I}^{\otimes (k-1)} \otimes \left( \mat{\bar D_k C_k D_k} \right)
 \otimes \mat{I}^{\otimes (p-k)} \\
 & = & \sum_{k=1}^{p} r_k \mat{I}^{\otimes (k-1)} \otimes (\mat{P_x})_k \otimes \mat{I}^{\otimes (p-k)}
 = \mat{\tilde H_x^{an}} \; .
\end{eqnarray*}
Analogously to $\mat{H_x}$ (see Eqn. \ref{eq:transformHx2Hz}),
the resulting matrix $\mat{\tilde H_x^{\text{an}}}$ can be diagonalized by the Kronecker product
$\mat{F_2} \otimes \cdots \otimes \mat{F_2}$.
Therefore, the eigenvalues of $\mat{H_x^{\text{an}}} + \mat{H_y^{\text{an}}}$ are given by all combinations
$$\pm r_1 \pm r_2 \pm \cdots \pm r_p \; .$$

Let us  return to analyzing the properties of Hamiltonians.
The matrix $\mat{H_z}$ is obviously diagonal and skew-persymmetric.
The matrix $\mat{H_{xx}}$ is again symmetric persymmetric (see Lemma \ref{lemma:KroneckerPersymm}).
Similar to $\mat{H_x}$, $\mat{H_{xx}}$ is again multilevel-circulant as it can be diagonalized by the
Kronecker product of the $2\times2$ Fourier matrix $\mat{F_2}$.
A computation similar to Eqn. \myref{eq:transformHx2Hz} results in
\begin{equation*}
 \left( \mat{F_2} \otimes \cdots \otimes  \mat{F_2}  \right) \left( \mat{H_{xx}} \right)
    (\mat{F_2} \otimes \cdots \otimes \mat{F_2})  = \mat{H_{zz}} \; .
\end{equation*}
The matrix $\mat{H_{yy}}$ is real symmetric persymmetric as becomes obvious from
\begin{equation*}
\mat{P_y} \otimes \mat{P_y} = \left(
                    \begin{array}{cc}
                      0 \ & \ -i \\
                      i \ & \ 0 \\
                    \end{array}
                  \right) \otimes \left(
                    \begin{array}{cc}
                      0 \ & \ -i \\
                      i \ & \ 0 \\
                    \end{array}
                  \right) = \left(
                              \begin{array}{cccc}
                                0 \ & \ 0 \ & \ 0 \ & \ -1 \\
                                0 \ & \ 0 \ & \ 1 \ & \ 0 \\
                                0 \ & \ 1 \ & \ 0 \ & \ 0 \\
                                -1 \ & \ 0 \ & \ 0 \ & \ 0 \\
                              \end{array}
                            \right) \;
\end{equation*}
being real and symmetric persymmetric,
which by Lemma \ref{lemma:KroneckerPersymm} translates into a real symmetric persymmetric matrix $\mat{H_{yy}}$.

The matrix $\mat{H_{zz}}$ is diagonal as it is constructed by a sum of Kronecker products of diagonal matrices.
Moreover $\mat{H_{zz}}$ is symmetric persymmetric via
\begin{equation*}
\mat{P_z} \otimes \mat{P_z} = \left(
                    \begin{array}{cc}
                      1 \ & \ 0 \\
                      0 \ & \ -1 \\
                    \end{array}
                  \right) \otimes \left(
                    \begin{array}{cc}
                      1 \ & \ 0 \\
                      0 \ & \ -1 \\
                    \end{array}
                  \right) = \left(
                              \begin{array}{cccc}
                                1 \ & \ 0 \ & \ 0 \ & \ 0 \\
                                0 \ & \ -1 \ & \ 0 \ & \ 0 \\
                                0 \ & \ 0 \ & \ -1 \ & \ 0 \\
                                0 \ & \ 0 \ & \ 0 \ & \ 1 \\
                              \end{array}
                            \right) \;
\end{equation*}
according to Remark \ref{rem:skewpersymm}.

Obviously, the spin-$1$ operators \myref{eq:spin1Operators} have similar symmetry properties
as their $2 \times 2$ counterparts:
the matrix $\mat{S_x}$ is real symmetric persymmetric and has Toeplitz format,
$\mat{S_y}/i$ is a real and skew-symmetric persymmetric Toeplitz matrix,
and $\mat{S_z}$ is symmetric skew-persymmetric and diagonal.
The Kronecker product $\mat{S_y} \otimes \mat{S_y}$ reads
\begin{equation*}
\mat{S_y} \otimes \mat{S_y}
= - \frac{1}{2} \pmat{0 \ & \ -1 \ & \ 0 \\ 1 \ & \ 0 \ & \ -1 \\ 0 \ & \ 1 \ & \ 0} \otimes
                \pmat{0 \ & \ -1 \ & \ 0 \\ 1 \ & \ 0 \ & \ -1 \\ 0 \ & \ 1 \ & \ 0} \; ,
\end{equation*}
a real symmetric persymmetric matrix (compare Remark \ref{rem:skewsymmetric}).
Following Remark \ref{rem:skewpersymm},
the Kronecker product $\mat{S_z} \otimes \mat{S_z}$ is symmetric persymmetric.
Therefore, according to Remark \ref{rem:powerspersymm} and Lemma \ref{lemma:KroneckerPersymm},
Both the AKLT model \myref{eq:defAKLT} and the generalized bilinear biquadratic model \myref{eq:defBilinearBiquadratic}
result in real symmetric persymmtric matrices.

Altogether all previously introduced physical models such as the 1D models listed in Table \ref{tab:Heisenberg}
define real and symmetric persymmetric matrices.
Due to Lemma~\ref{lemma:EigenvectorsPersymmMatrices},
the related eigenvectors such as the ground state (which corresponds to the lowest-lying eigenvalue)
are either symmetric or skew-symmetric.

\section{Application to Matrix Product States}
For efficiently simulating quantum many-body systems, one has to find a sparse (approximate)
representation, because otherwise the state space would
grow exponentially with the number of particles.
Here \/`efficiently\/' means using resources (and hence representations)
growing only polynomially in the system size~$p$.
In the quantum information (QI) society, Matrix Product States are in use to treat 1D problems.

\subsection{Matrix Product States: Formalism and Normal Forms}
This paragraph summarizes some well-known basics about \mps.
We provide both the \mps formalism and normal forms for \mps,
which are well-known in the QI society, from a (multi-)linear algebra point of view.
Afterwards we present own findings to construct normal forms and discuss the benefit of such forms.

\subsubsection{Formalism}

For 1D spin systems, consider Matrix Product States,
where every physical site $j$ is associated with a pair of matrices
$\mat{A_j^{(0)}}, \mat{A_j^{(1)}} \in \mathbb C^{D_j \times D_{j+1}}$,
representing one of the two possibilities spin-up or spin-down.

Let $(i_1,i_2,\dots,i_p)$ denote the binary representation of the integer index $i$.
Then the $i$th vector component takes the form
\begin{equation}\label{eq:MPSComponent}
x_i = x_{i_1,\dots, i_p} = \trace \left(\mat{A_1^{(i_1)}}\cdot \mat{A_2^{(i_2)}}\cdots \mat{A_p^{(i_p)}} \right) \; .
\end{equation}
Hence, the overall vector $\vec{x}$ can be expressed as
\begin{equation*}\label{eq:MPSVector}
\begin{split}
\vec{x} & = \sum\limits_{i=1}^{2^p} x_i \vec{e_i} =
                \sum\limits_{i_1,i_2,\dots,i_p} x_{i_1,\dots, i_p} \vec{e_{i_1}} \otimes \cdots \otimes \vec{e_{i_p}} \\
        & = \sum\limits_{i_1,\dots, i_p} \trace
                \left(\mat{A_1^{(i_1)}}\cdot \mat{A_2^{(i_2)}}\cdots \mat{A_p^{(i_p)}} \right)
                \vec{e_{i_1}} \otimes \cdots \otimes \vec{e_{i_p}}\\
        & = \sum\limits_{i_1,\dots, i_p} \biggl(
                \sum\limits_{m_1,\dots, m_p} A_{1;m_1,m_2}^{(i_1)}\cdot A_{2;m_2,m_3}^{(i_2)}\cdots A_{p;m_p,m_1}^{(i_p)}
                \biggr) \vec{e_{i_1}} \otimes \cdots \otimes \vec{e_{i_p}} \\
        & = \sum\limits_{m_1,\dots, m_p} \biggl( \sum_{i_1} A_{1;m_1,m_2}^{(i_1)} \vec{e_{i_1}} \biggr)
        \otimes \cdots \otimes \biggl( \sum_{i_p} A_{p;m_p,m_1}^{(i_p)} \vec{e_{i_p}} \biggr) \\
        & = \sum_{m_1,m_2,...,m_p} \vec{a_{1;m_1,m_2}} \otimes \vec{a_{2;m_2,m_3}} \otimes
\cdots\otimes   \vec{a_{p;m_p,m_1}}
\end{split}
\end{equation*}
with vectors $\vec{a_{j;m_j,m_{j+1\; \mathrm{mod}\; p}}}$ of length 2. These vectors are
pairs of entries at position $m_j,m_{j+1\; mod\; p}$ from the matrix pair
$\mat{A_j^{(i_j)}}$, $i_j=0,1$.

We distinguish between open boundary conditions, where $D_1 = D_{p+1} = 1$
and periodic boundary conditions,
where the first and last particles are also connected: $D_1=D_{p+1}>1$.
The first case corresponds to the Tensor Train format (\cite{Oseledets11tt}),
the latter to the Tensor Chain format(\cite{Khoromskij11QuanticsApproximation}).
Considerations on {\sc{mps}} from a mathematical point of view can be found in \cite{Huckle11Computations}.

\subsubsection{Normal Forms} \label{sec:normalFormsMPS}
The \MPS ansatz does not lead to unique representations,
because we can always introduce factors of the form $\mat{M_j}\mat{M_j^{\text{$-1$}}}$ between $\mat{A_j^{(i_j)}}$ and
$\mat{A_{j+1}^{(i_{j+1})}}$.
In order to reduce this ambiguity in the open boundary case one can use the SVD
to replace the matrix pair $(\mat{A_j^{(0)}},\mat{A_j^{(1)}})$ by parts of unitary matrices
(see, e.g. \cite{SchollDMRG2011}).
To this end, one may start from the left (right), carry out an SVD,
replace the current pair of \mps matrices by parts of unitary matrices,
shift the remaining SVD part to the right (left) neighbor, and proceed recursively with the neighboring site.
Starting from the left one obtains a \defini{left-normalized} \mps representation fulfilling the gauge condition
\begin{equation}\label{eq:mpsGaugeLeft}
\bigl(\mat{A_j^{(0)}} \bigr)^{\herm}\mat{A_j^{(0)}} + \bigl( \mat{A_j^{(1)}} \bigr)^{\herm}\mat{A_j^{(1)}}=\mat{I} \; .
\end{equation}
Analogously, if we start the procedure from the right, we end up with a \defini{right-normalized} \mps representation fulfilling
\begin{equation}\label{eq:mpsGaugeRight}
\mat{A_j^{(0)}} \bigl(\mat{A_j^{(0)}} \bigr)^{\herm} +
\mat{A_j^{(1)}} \bigl( \mat{A_j^{(1)}} \bigr)^{\herm} = \mat{I} \; .
\end{equation}
In the periodic boundary case these gauge conditions can only be achieved for all up to one site.

Still some ambiguity remains because we can insert $\mat{W_j W_j}^{\herm}$ with any unitary $\mat{W_j}$
in the \MPS representation \myref{eq:MPSComponent}
between the two terms at position $j$ and $j+1$
without any effect to the gauge conditions \myref{eq:mpsGaugeLeft} or \myref{eq:mpsGaugeRight}.
To overcome this ambiguity a stronger normalization can be derived (see, e.g. \cite{eckholt}).
It is based on different matricizations of the vector to be represented and can be written in the form
\begin{equation}\label{eq:mpsRepresentationLambdaGamma}
\vec{x_{i_1...i_p}} =
\mat{\Gamma_1^{(i_1)}} \bigl( \mat{\Lambda_1\Gamma_2^{(i_2)}} \bigr) \bigl( \mat{ \Lambda_2 \Gamma_3^{(i_3)}} \bigr)\cdots
\bigl( \mat{\Lambda_{p-1}\Gamma_p^{(i_p)}} \bigr) =
\mat{A_1^{(i_1)} A_2^{(i_2)} A_3^{(i_3)}} \cdots \mat{A_p^{(i_p)}}
\end{equation}
with diagonal matrices $\mat{\Lambda_j}$ containing the singular values of special
matricizations of the vector $\vec{x}$.
The following lemma states the existence of such an \mps representation.
\begin{lemma}[(\cite{Vidal03}) ]\label{lemma:ExistenceUniqueMPSGauge}
Any vector $\vec{x} \in \mathbb C^{2^p}$ of norm $1$ can be represented by an \mps representation
fulfilling the left conditions
\begin{subequations}\label{eq:MPS2normalizationConditionsFromLeft}
\begin{align} \label{eq:MPS2normalizationConditionsFromLeftConda}
  \bigl( \mat{A_j^{(0)}} \bigr)^{\herm} \mat{A_j^{(0)}} + \bigl( \mat{A_j^{(1)}} \bigr)^{\herm}\mat{A_j^{(1)}}
    & =\mat{I} \\ \label{eq:MPS2normalizationConditionsFromLeftCondb}
  \mat{A_j^{(0)}\Lambda_j^{\text{$2$}}} \bigl(\mat{A_j^{(0)}}\bigr)^{\herm} +
  \mat{A_j^{(1)}\Lambda_j^{\text{$2$}}} \bigl( \mat{A_j^{(1)}} \bigr)^{\herm}
    & = \mat{\Lambda}_{\mat{j-1}}^2
\end{align}
\end{subequations}
or the right conditions
\begin{subequations}\label{eq:MPS2normalizationConditionsFromRight}
\begin{align}
\mat{A_j^{(0)}}\bigl(\mat{A_j^{(0)}}\bigr)^{\herm}+\mat{A_j^{(1)}}\bigl(\mat{A_j^{(1)}}\bigr)^{\herm} & =\mat{I} \\
\bigl(\mat{A_j^{(0)}}\bigr)^{\herm}\mat{\Lambda}_{\mat{j-1}}^2 \mat{A_j^{(0)}} +
\bigl( \mat{A_j^{(1)}}\bigr)^{\herm}\mat{\Lambda}_{\mat{j-1}}^2 \mat{A_j^{(1)}} & =
        \mat{\Lambda}_{\mat{j}}^2 \; ,
\end{align}
\end{subequations}
where the $D_{j+1}\times D_{j+1}$ diagonal matrices $\mat{\Lambda_j}$ contain the non-zero singular values of the
matricization of $\vec{x}$ relative to index partitioning $(i_1,...,i_j),(i_{j+1},...,i_p)$,
diagonal entries ordered in descending order.
\end{lemma}
The proof of this lemma is constructive and provides \mps factors~$\mat{A_j^{(i_j)}}$ again
as parts of unitary matrices, but satisfying {\bf two} normalization conditions.
These conditions are well-known in the QI society,
see, e.g., \cite{Vidal03, Verstraete06}.
The following proof is adapted from \cite{eckholt}, but we reformulate it in mathematical (matrix) notation.
\begin{proof}
Let us prove representation \myref{eq:MPS2normalizationConditionsFromLeft}
for a given vector $\vec{x}$ by orthogonalization from the left.
We start with considering the SVD of the first matricization relative to $i_1$,
\begin{equation}\label{eq:FirstMatricastion}
\mat{X_{i_1,(i_2,...,i_p)}}=\mat{U_1 \Lambda_1 W_2}=
\left(\begin{array}{c}
\mat{A_1^{(0)}\Lambda_1 W_2}\\
\mat{A_1^{(1)}\Lambda_1 W_2}
\end{array}
\right) = \left(
            \begin{array}{c}
              \mat{\Gamma_1^{(0)}\Lambda_1 W_2} \\
              \mat{\Gamma_1^{(1)}\Lambda_1 W_2} \\
            \end{array}
          \right)
\end{equation}
with the notation $\mat{A_1}=\mat{U_1}=\mat{\Gamma_1}$ and $\mat{\Lambda_1}$ containing all positive singular values.
Therefore, the columns of $\mat{U_1}$ are
pairwise orthonormal satisfying
$$\mat{I}=\mat{A_1^{\herm} A_1} =
    \bigl( \mat{A_1^{(0)}} \bigr)^{\herm}\mat{A_1^{(0)}}+\bigl(\mat{A_1^{(1)}}\bigr)^{\herm}\mat{A_1^{(1)}}\; .$$
Now, the second matricization gives the SVD
\begin{equation}\label{eq:SecondMatricasation}
\mat{X_{(i_1,i_2),(i_3,...,i_p)}}=\mat{U_2 \Lambda_2 W_3}=
\left(
\begin{array}{c}
\mat{U_2^{(0)}}\\
\mat{U_2^{(1)}}
\end{array}
\right)
\mat{\Lambda_2 W_3}\; .
\end{equation}
Note that because both matricizations \myref{eq:FirstMatricastion} and \myref{eq:SecondMatricasation}
represent the same vector $\vec X$, each column of $\mat{U_2^{(0)}}$ can be
represented as $\mat{\Gamma_1\Lambda_1}\cdot \mat{\Gamma_2^{(0)}}$ for some $\mat{\Gamma_2^{(0)}}$.
This follows by
$$\mat{U_2^{(0)}\Lambda_2 W_3}= \mat{\Gamma_1 \Lambda_1 W_2}\; .$$
Picking a full-rank submatrix $\mat{C}$ of $\mat{\Lambda_2 W_3}$ and applying the inverse from the right yields
$$\mat{U_2^{(0)}}=\mat{\Gamma_1 \Lambda_1 \hat W_2} \; .$$
The same holds for $\mat{U_2^{(1)}}$ with some $\mat{\Gamma_2^{(1)}}$.
With these matrices $\mat{\Gamma_2^{(0)}}$ and $\mat{\Gamma_2^{(1)}}$ we can write
\begin{equation*}
\mat{U_2}=\pmat{
\mat{U_2^{(0)}}\\
\mat{U_2^{(1)}}
}=
\pmat{
\mat{\Gamma_1\Lambda_1\Gamma_2^{(0)}}\\
\mat{\Gamma_1\Lambda_1\Gamma_2^{(1)}}
}=
\pmat{
\mat{\Gamma_1 A_2^{(0)}}\\
\mat{\Gamma_1 A_2^{(1)}}}
\end{equation*}
with
\begin{equation}\label{eq:definitionA_2}
\pmat{
\mat{A_2^{(0)}}\\
\mat{A_2^{(1)}}
}:=
\pmat{
\mat{\Lambda_1 \Gamma_2^{(0)}}\\
\mat{\Lambda_1 \Gamma_2^{(1)}}
} =
\pmat{
\mat{\Gamma_1^{\herm} \Gamma_1\Lambda_1 \Gamma_2^{(0)}}\\
\mat{\Gamma_1^{\herm} \Gamma_1 \Lambda_1 \Gamma_2^{(1)}}
}=
\pmat{
\mat{\Gamma_1^{\herm}U_2^{(0)}}\\
\mat{\Gamma_1^{\herm}U_2^{(1)}}
}\; .
\end{equation}
In view of the SVD representation \myref{eq:SecondMatricasation} of $\mat{X_{(i_1,i_2),(i_3,...,i_p)}}$ one finds
\begin{eqnarray*}
\mat{I}  = \mat{U_2^{\herm}}\mat{U_2} & = & \bigl(\mat{A_2^{(0)}}\bigr)^{\herm} \mat{\Gamma_1^{\herm}}
                \mat{\Gamma_1 }\mat{A_2^{(0)}}+\bigl(\mat{A_2^{(1)}}\bigr)^{\herm} \mat{\Gamma_1^{\herm}}
                \mat{\Gamma_1 A_2^{(1)}} \\
& = & \bigl(\mat{A_2^{(0)}}\bigr)^{\herm}\mat{A_2^{(0)}}+\bigl(\mat{A_2^{(1)}}\bigr)^{\herm}\mat{A_2^{(1)}}\; ,
\end{eqnarray*}
which corresponds to the first normalization condition \myref{eq:MPS2normalizationConditionsFromLeftConda}.
Now we can rewrite the second matricization \myref{eq:SecondMatricasation} as
\begin{equation*}
\mat{X_{(i_1,i_2),(i_3,...,i_p)}}=
\pmat{
\mat{U_2^{(0)}}\\
\mat{U_2^{(1)}}
}\mat{ \Lambda_2 W_3}=
\pmat{
\mat{\Gamma_1 \Lambda_1 \Gamma_2^{(0)}\Lambda_2 W_3}\\
\mat{\Gamma_1 \Lambda_1 \Gamma_2^{(1)}\Lambda_2 W_3}
}\; .
\end{equation*}
Comparing this form of the vector $\vec X$ with the first matricization \myref{eq:FirstMatricastion} gives
\begin{equation*}
\mat{W_2}=\pmat{
\mat{\Gamma_2^{(0)}\Lambda_2 W_3} \ & \ \mat{\Gamma_2^{(1)}\Lambda_2 W_3}
}
\end{equation*}
and therefore
{\allowdisplaybreaks\begin{align*}
\mat{I} & =\mat{W_2W_2^{\herm}} =
\pmat{
\mat{\Gamma_2^{(0)}\Lambda_2 W_3} \ & \ \mat{\Gamma_2^{(1)}\Lambda_2 W_3}
}
\pmat{ \mat{W_3^{\herm} \Lambda_2} \bigl( \mat{\Gamma_2^{(0)}} \bigr)^{\herm} \\
\mat{W_3^{\herm} \Lambda_2} \bigl( \mat{\Gamma_2^{(1)}} \bigr)^{\herm}
} \\
& =
\mat{\Gamma_2^{(0)} \Lambda_2^{\text{$2$}}} \bigl( \mat{ \Gamma_2^{(0)} } \bigr)^{\herm} +
\mat{\Gamma_2^{(1)} \Lambda_2^{\text{$2$}}} \bigl( \mat{ \Gamma_2^{(1)} } \bigr)^{\herm} \; .
\end{align*}}
Multiplying from both sides with $\mat{\Lambda_1}$ is just the second
condition~\myref{eq:MPS2normalizationConditionsFromLeftCondb}:
\begin{equation*}
\begin{split}
\mat{\Lambda}_{\mat{1}}^2 & =
\mat{\Lambda_1 \Gamma_2^{(0)} \Lambda_2^{\text{$2$}}} \bigl( \mat{ \Gamma_2^{(0)} } \bigr)^{\herm} \mat{\Lambda_1} +
\mat{\Lambda_1 \Gamma_2^{(1)} \Lambda_2^{\text{$2$}}} \bigl( \mat{ \Gamma_2^{(1)} } \bigr)^{\herm} \mat{\Lambda_1} \\
& =
\mat{A_2^{(0)} \Lambda_2^{\text{$2$}}} \bigl( \mat{ A_2^{(0)} } \bigr)^{\herm} +
\mat{A_2^{(1)} \Lambda_2^{\text{$2$}}} \bigl( \mat{ A_2^{(1)} } \bigr)^{\herm} \; .
\end{split}
\end{equation*}

In the same way we can use the two matricizations $\mat{X_{(i_1,i_2),(i_3,...,i_p)}}$ and
$\mat{X_{(i_1,i_2,i_3),(i_4,...,i_p)}}$ to derive $\mat{A_3}$,
based on $\mat{\Lambda_2}$, $\mat{\Gamma_2}$, $\mat{U_3}$, $\mat{W_3}$, $\mat{W_4}$
and $\mat{\Lambda_3}$,
satisfying the normalization conditions \myref{eq:MPS2normalizationConditionsFromLeft}.
Then $\mat{A_4},...,\mat{A_p}$ follow similarly.

Starting from the right and using a similar procedure
gives the representation satisfying the normalization conditions \myref{eq:MPS2normalizationConditionsFromRight}.
\end{proof}
\begin{remark}
\begin{enumerate}
\item The resulting \mps representation is unique up to unitary diagonal matrices as long as the
singular values in each diagonal matrix $\mat{\Lambda_j}$ are in descending order
and have no degeneracy (are all different), compare \cite{PerezGarcia07MPS}.
\item One may consider the constructive proof as a possible introduction of \mps (\cite{Vidal03,SchollDMRG2011}).
Then the conditions \myref{eq:MPS2normalizationConditionsFromLeft} or
\myref{eq:MPS2normalizationConditionsFromRight} appear naturally.
\item The proof shows that, in general,
an exact representation comes at the cost of exponentially growing matrix dimensions $D_j$.
For keeping the matrix dimensions limited one would have to introduce SVD-based truncations.
\item The Vidal normalization~\cite{Vidal03} uses
$\mat{\Gamma_j^{(i_j)}}$ and $\mat{\Lambda_j}$ in \myref{eq:mpsRepresentationLambdaGamma}
instead of $\mat{A_j^{(i_j)}}$.
\item Starting from a given \mps $\mat A$-representation~\myref{eq:MPSComponent} it is possible (\cite{SchollDMRG2011})
to build an equivalent $\mat{\Lambda \Gamma}$-representation~\myref{eq:mpsRepresentationLambdaGamma}
without considering the matricizations explicitly.
The construction starts from a right-normalized \mps representation~\myref{eq:mpsGaugeRight}
and then iteratively computes SVDs of modified decompositions related to two neighboring sites.
The conversion from the $\mat{\Lambda \Gamma}$-form to the $\mat A$-form is simpler:
From \myref{eq:definitionA_2} it becomes obvious to set $\mat{A_j^{(i_j)}}=\mat{\Lambda_{j-1} \Gamma_j^{(i_j)}}$
($\mat{\Lambda_0} := 1$) in the left-normalized case \myref{eq:MPS2normalizationConditionsFromLeft}.
Analogously, in the right normalized case \myref{eq:MPS2normalizationConditionsFromRight}
we would define $\mat{A_j^{(i_j)}}=\mat{\Gamma_j^{(i_j)} \Lambda_{j}}$, where $\mat{\Lambda_p} := 1$.
\item The $\mat{\Lambda \Gamma}$-representation \myref{eq:mpsRepresentationLambdaGamma}
corresponds to the Schmidt decomposition,
which is well-known in QI.
The Schmidt coefficients are just the diagonal entries of $\mat{\Lambda_j}$ (\cite{SchollDMRG2011}).
\item The diagonal matrices $\mat{\Lambda_j}$ contain the singular values of special matricizations
of the vector to be represented. Hence, local matrices $\mat{A_j}$ reflect global information on the tensor
via the normalization conditions and the diagonal matrices $\mat{\Lambda_j}$.
That is one of the reasons why \mps has proper approximation properties (\cite{Verstraete06}).
\end{enumerate}
\end{remark}

\subsubsection{Further Normal Forms}
Finally we propose own findings of concepts to introduce possible normal forms for \mps.

As an alternative to construct the gauge conditions \myref{eq:mpsGaugeLeft} or \myref{eq:mpsGaugeRight}
we propose (compare \cite{Huckle11Computations}) to consider two neighboring pairs
(compare two-site \DMRG \cite{SchollDMRG2011})
\begin{eqnarray}\nonumber
 \left(
  \begin{array}{c}
        \mat{A_j^{(0)}} \\
        \mat{A_j^{(1)}} \\
  \end{array}
 \right)  \cdot
 \left(
  \begin{array}{cc}
        \mat{A_{j+1}^{(0)}} \ & \ \mat{A_{j+1}^{(1)}} \\
  \end{array}
 \right) & = &
\left(
  \begin{array}{cc}
        \mat{A_j^{(0)}A_{j+1}^{(0)}} \ & \ \mat{A_j^{(0)}A_{j+1}^{(1)}} \\
        \mat{A_j^{(1)}A_{j+1}^{(0)}} \ & \ \mat{A_j^{(1)}A_{j+1}^{(1)}} \\
  \end{array}
 \right) \stackrel{\rm SVD}{=}  \\ \label{eq:dmrgSVDLeft}
 \left(
  \begin{array}{c}
        \mat{U_j^{(0)}} \\
        \mat{U_j^{(1)}} \\
  \end{array}
 \right)  \mat{\Lambda_j}
 \left(
  \begin{array}{cc}
        \mat{U_{j+1}^{(0)}} \ & \ \mat{U_{j+1}^{(1)}} \\
  \end{array}
 \right)
  & = & \left(
  \begin{array}{c}
        \mat{U_j^{(0)}} \\
        \mat{U_j^{(1)}} \\
  \end{array}
 \right) \left(
           \begin{array}{cc}
             \mat{\Lambda_j U_{j+1}^{(0)}} \ & \ \mat{\Lambda_j} \mat{U_{j+1}^{(1)}} \\
           \end{array}
         \right)\\ \label{eq:dmrgSVDRight}
    & = & \left(
  \begin{array}{c}
        \mat{U_j^{(0)} \Lambda_j} \\
        \mat{U_j^{(1)} \Lambda_j} \\
  \end{array}
 \right) \left(
           \begin{array}{cc}
             \mat{U_{j+1}^{(0)}} \ & \ \mat{U_{j+1}^{(1)}} \\
           \end{array}
         \right)
     \; .
 \end{eqnarray}
In this way all matrix pairs $(\mat{A_j^{(0)}},\mat{A_j^{(1)}})$ (up to one in the periodic boundary case)
can be assumed as part of a unitary
matrix giving the normalization conditions \myref{eq:mpsGaugeLeft}
in the left-normalized case \myref{eq:dmrgSVDLeft} or \myref{eq:mpsGaugeRight}
in the right-normalized case \myref{eq:dmrgSVDRight}.

To circumvent the fact that the gauge conditions \myref{eq:mpsGaugeLeft} or \myref{eq:mpsGaugeRight}
still introduce some ambiguity
we propose the following way to derive a stronger normalization.
Suppose that the \mps matrices are already in the left-normalized form
$$ \Bigl( \mat{A_j^{(0)}} \Bigr)^{\herm} \mat{A_j^{(0)}} +
\Bigl( \mat{A_j^{(1)}} \Bigr)^{\herm} \mat{A_j^{(1)} } = \mat{I} \qquad \text{for} \ j=1,\dots,p \; .$$
The proposed normal form is now based on the SVD of the
upper matrices \mbox{$\mat{A_j^{(0)}}=\mat{U_j\Sigma_jV_j}$} with unitary $\mat{U_j,V_j}$ ($\mat{V_0}:=1$) and
diagonal non-negative $\mat{\Sigma_j}$, diagonal entries ordered relative to absolute value.
Then, every pair $\bigl(\mat{A_j^{(0)}},\mat{A_j^{(1)}}\bigr)$ is replaced by
\begin{equation}\label{eq:mpsConstructFurtherNormalForm}
\left(\mat{\tilde A_j^{(0)}} \; , \; \mat{\tilde A_j^{(1)}}\right) =
        \left( \mat{V_{j-1}U_j\Sigma_j} \; , \; \mat{V_{j-1}A_j^{(1)}V_j^{\herm}} \right)
\end{equation}
leading to the stronger normalization conditions
\begin{equation}\label{eq:mpsFurtherNormalForm}
\bigl(\mat{\tilde A_j^{(0)}}\bigr)^{\herm}\mat{\tilde A_j^{(0)}}+\bigl(\mat{\tilde A_j^{(1)}}\bigr)^{\herm}
\mat{\tilde A_j^{(1)}}=\mat{\Sigma_j^{\herm}}\mat{\Sigma_j}+\mat{\Delta_j^{\herm}}\mat{\Delta_j}=\mat{I}
\end{equation}
with diagonal matrices $\mat{\Sigma_j}$ and $\mat{\Delta_j}$.
From \myref{eq:mpsFurtherNormalForm} we can read that this normal form provides \mps matrices with orthogonal columns.
For the upper matrices $\mat{\tilde A_j^{(0)}}$ this fact is caused by construction,
but it then automatically follows also for the $\mat{\tilde A_j^{(1)}}$ matrices.
Especially for the left-most site $j=1$, the normalization condition~\myref{eq:mpsFurtherNormalForm}
leads to $\mat{\tilde A_1^{(0)}}=(1,0)$ and $\mat{\tilde A_1^{(1)}}=(0,1)$.
We may of course also start the proposed normalization procedure with a right-normalized form,
resulting in a representation where the \mps matrices have orthogonal rows.

\subsubsection{Comparison of the Normal Forms}
All of the presented normal forms introduce some kind of uniqueness to the \mps formalism,
which initially is not unique.
Therefore, these normal forms help to prevent redundancy in the representations.
As a consequence we may expect less memory demands as well as better properties of numerical algorithms
such as faster convergence, better approximation, and improved stability.
The normal form~\myref{eq:MPS2normalizationConditionsFromLeft}
is advantageous as it connects local and global information.
However, the construction involves the inverse of the diagonal SVD matrices which may cause numerical problems.
Our normal form \myref{eq:mpsFurtherNormalForm} can be built without division by singular values,
but the information is more local.

\subsection{Symmetries in \MPS}
The results from Section~\ref{sec:structuredMatrices} show that the matrices which describe the physical model systems
have special symmetry properties which result in symmetry properties of the related eigenvectors:
the eigenvector of a symmetric persymmetric Hamiltonian has to be symmetric or skew-symmetric,
i.e. $\mat J \vec v = \pm \vec v$.
One might also think about other symmetries which could be of the form
$$ \vec{v} = \left(
               \begin{array}{c}
                 \vec a \\
                 \vec a \\
               \end{array}
             \right) \; , \ \vec{v} = \left(
                                        \begin{array}{c}
                                          \vec a \\
                                          - \vec a \\
                                        \end{array}
                                      \right) \; ,
\quad\text{or more general}\quad \mat{P} \vec{v}=\pm \vec{v}
$$
with a general permutation $\mat{P}$. Furthermore, we can
have vectors satisfying $k$ different independent symmetry properties, e.g.
$\mat{P_j} \vec{v}=\pm \vec{v}$ for permutations $\mat{P_j}$, $j=1,...,k$.

At this point the question arises how these symmetry properties can be
expressed in terms of \MPS,
and, vice versa, how special properties such as certain relations between the \MPS matrices
emerge in the represented vector.

Symmetries in \mps already appear in different QI publications:
theoretical considerations on symmetries in \MPS can be found in \cite{PWSVC08String,SWPC09MPS},
symmetries in \TIMPS representations are exploited in \cite{Pirvu11ExploitingTI},
and the application of involutions has been analyzed in \cite{Sandvik07Variational}.
The main goal of this paragraph is to present an overview of different types of symmetries in a unifying way
and to give results concerning the uniqueness of such symmetry-adapted representation approaches
by proposing possible normal forms.
Our results are intended for a theoretical purpose (similar to \cite{PWSVC08String,SWPC09MPS})
but are also interesting for numerical applications (similar to \cite{Sandvik07Variational,Pirvu11ExploitingTI}).

After some technical considerations we discuss which properties of the matrices $\mat{A_j^{(i_j)}}$
that define an \MPS vector $\vec{x}$ are related to certain symmetry properties
of $\vec{x}$. Deriving normal forms
for different symmetries of \MPS vectors will also be of interest.

\subsubsection{Technical Remarks}
In view of the trace taken in the \MPS formalism \myref{eq:MPSComponent},
recall the following trivial yet useful properties 
{\allowdisplaybreaks
\begin{align}\label{eq:traceInvariance}
    \trace \left( \mat{A B} \right) & = \trace \left( \mat{B A} \right) \; , && \\ \label{eq:traceTranspose}
    \trace \left( \mat{A B} \right) & =  \trace \left( \mat{A B }\right)^{\T} =
	\trace \left( \mat{B}^{\T} \mat{A}^{\T} \right)
        && \text{ for } \trace \left( \mat{A} \mat{B} \right) \in \mathbb R \; , \\ \label{eq:traceHermitian}
\trace \left( \mat{A B} \right) & =  \overline{\trace \left( \mat{A B} \right)^{\herm}} =
	\overline{\trace \left( \mat{B}^{\herm} \mat{A}^{\herm} \right)}
        && \text{ for } \trace \left( \mat{A} \mat{B} \right) \in \mathbb C \;
\end{align}}
in order to arrive at relations of the form
{\allowdisplaybreaks
\begin{align*}
    \trace\left( \mat{A_1^{(i_1)}}\cdot \mat{A_2^{(i_2)}}\cdots \mat{A_p^{(i_p)}} \right) & \stackrel{\myref{eq:traceInvariance}}{=}
    \trace\left( \mat{A_{r+1}^{(i_{r+1})}} \cdots \mat{A_p^{(i_p)} A_1^{(i_1)}} \cdots \mat{A_r^{(i_r)}} \right) \\
    & \stackrel{\myref{eq:traceHermitian}}{=} \overline{
        \trace \left( \mat{A_r^{(i_r) \herm} A_{r-1}^{(i_{r-1}) \herm}} \cdots \mat{A_1^{(i_1) \herm} A_p^{(i_p) \herm}}
        \cdots \mat{A_{r+1}^{(i_{r+1}) \herm}}\right) } \; .
\end{align*}}
For the proof of the main theorems we will need the following three lemmata.
\begin{lemma} \label{lemmaTraceEqual}
Let $A,B \in \mathbb K^{n \times m}$, where $\mathbb K \in \{ \mathbb R, \mathbb C\}$.
If the equality
$$ \trace \left( \mat{A X} \right) = \trace \left( \mat{B X} \right)$$
holds for all matrices $\mat X \in \mathbb K^{m \times n}$, then $\mat{A} = \mat{B}$.
\end{lemma}
\begin{proof}
The relation $\trace \left( \mat{A X} \right) = \trace \left( \mat{B X} \right)$ is equivalent to
$$ \trace \left( (\mat{A} - \mat{B}) \mat{X} \right) = 0 $$
for all matrices $\mat{X}$.
For the special choice $\mat{X} = (\mat{A}-\mat{B})^{\herm}$ we obtain
$$ \trace \left( (\mat{A} - \mat{B}) (\mat{A}-\mat{B})^{\herm} \right) = \| \mat{A} - \mat{B} \|_{\operatorname F}^2 = 0 \; ,$$
which shows $\mat{A} = \mat{B}$.
\end{proof}

\begin{lemma}\label{LemmaCommute1}
Assume that for $\mat{U} \in \mathbb K^{n \times n}$ and
$\mat{V} \in \mathbb K^{m \times m}$ it holds
\begin{equation}\label{eq:ConditionLemmaCommute1}
\mat{X}  = \mat{V X U}
\end{equation}
for all matrices $\mat{X} \in \mathbb K^{m \times n}$. Then $\mat{U} = c \mat{I_n}$, $\mat{V}=\mat{I_m}/c$
with some $c\neq 0$.
\end{lemma}

\begin{proof}
Obviously, $\mat{V}$ and $\mat{U}$ have to be non-zero and, moreover,
they are regular.
Otherwise, if e.g.
$\mat{V} \vec a=\vec 0$ for $\vec a \neq \vec 0$,
we can define $\mat{X}=\vec a \vec b^{\herm}$ with some
$\vec b \neq \vec 0$ leading to a contradiction.
Choosing $\mat{X}= \vec a \vec b^{\herm}$ as rank-one matrix for any vectors $\vec a$ and $\vec b$,
it follows
$(\mat{V}^{-1}\vec a) \vec b^{\herm} = \vec a( \vec b^{\herm}\mat{U})$.
Therefore, $\mat{V}^{-1} \vec a$ and $\vec a$
have to be collinear ($\mat{V}^{-1} \vec a = \lambda \vec a$ with some $\lambda \in \mathbb K$),
and $\vec b^{\herm}\mat{U}$ and $\vec b^{\herm}$ also have to
be collinear ($\vec b^{\herm} \mat{U} = \mu \vec b^{\herm}$).
Hence, $\mat U$ and $\mat V^{-1}$ (and therefore also $\mat{V}$) have all vectors of appropriate size as eigenvectors,
and therefore they are nonzero multiples of the identity matrix, $\mat{U} = c_1 \mat{I_n}$ and $\mat{V} = c_2 \mat{I_m}$.
Condition \myref{eq:ConditionLemmaCommute1} finally shows $c = c_1 = 1/c_2$.
\end{proof}

Similarly, we can derive the following result:

\begin{lemma}\label{LemmaCommute2}
Assume that for $\mat{U} \in \mathbb K^{n \times n}$ and
$\mat{V} \in \mathbb K^{m \times m}$ it holds
$$ \mat{X U} = \mat{V X} $$
for all matrices $\mat{X} \in \mathbb K^{m \times n}$.
Then $\mat{U} = c \mat{I_n}$, $\mat{V}=c \mat{I_m}$ with a scalar $c \in \mathbb K$.
\end{lemma}

\begin{proof}
First we prove that, if at least one of the two matrices $\mat U$ or $\mat V$ is singular,
both of them have to be zero.
Obviously, if one of the two matrices is zero, the other one has to be zero as well.
If we now suppose $\mat V$ to be singular and $\mat U$ to be nonzero,
we can find vectors $\vec a \not= \vec{0}$ and $\vec b$, such that $\mat V \vec a = \vec{0}$ and
$\vec b^{\herm} \mat U \not= \vec 0^{\herm}$.
The choice $\mat X = \vec a \vec b^{\herm} \not= \mat 0$ leads to a contradiction.
The same argument counts if we change the roles of $\mat U$ and $\mat V$.

Otherwise, if both matrices are regular, the statement of the lemma is a direct consequence
from Lemma~\ref{LemmaCommute1}.
\end{proof}

%

\subsubsection{Bit-Shift Symmetry and Translational Invariance}
To begin, consider the case where all matrix pairs are equal,
i.e.
\begin{equation}\label{eq:defTIMPS}
 \left(
     \begin{array}{c}
       \mat{A_j^{(0)}} \\
       \mat{A_j^{(1)}} \\
     \end{array}
   \right) = \left(
               \begin{array}{c}
                 \mat{A^{(0)}} \\
                 \mat{A^{(1)}} \\
               \end{array}
             \right)
\end{equation}
for all $j=1,\dots,p$.
Then the \MPS is site-independent and describes a \defini{translational invariant} (TI) state
on a spin system with periodic boundary conditions \cite{PerezGarcia07MPS}.
The following theorem states that the result of such a relation is a \defini{bit-shift symmetry}, i.e.
$$ x_{i_1,i_2,\dots,i_p} = x_{i_2,i_3,\dots,i_p,i_1}  = \dots = x_{i_p,i_1,i_2,\cdots, i_{p-1}} \; .$$
\begin{theorem}[(\cite{PerezGarcia07MPS}) ]
If the \mps matrices are site-independent (and thus fulfill Eqn.~\ref{eq:defTIMPS})
the represented vector has the bit-shift symmetry and in turn every vector with the bit-shift symmetry
can be represented by a site-independent \mps.
\end{theorem}

\begin{proof}
To see that a \TIMPS~\myref{eq:defTIMPS} leads to a bit-shift symmetry, consider
\begin{eqnarray*}
x_{i_1,i_2,\dots,i_p} & = & \trace \left( \mat{A_1^{(i_1)} A_2^{(i_2)}} \cdots \mat{A_p^{(i_p)}} \right) \\
        & \stackrel{\myref{eq:traceInvariance}}{=} & \trace \left( \mat{A^{(i_2)} A^{(i_3)}}
            \cdots \mat{A^{(i_p)} A^{(i_1)}} \right) = x_{i_2,i_3,\dots,i_p,i_1} \\
        & \stackrel{\myref{eq:traceInvariance}}{=} & \trace
        \left( \mat{A^{(i_3)} A^{(i_4)}} \cdots \mat{A^{(i_p)} A^{(i_1)} A^{(i_2)}} \right) = x_{i_3,\dots,i_p,i_1,i_2} \\
        & = & \cdots \\
        & \stackrel{\myref{eq:traceInvariance}}{=} & \trace \left( \mat{A^{(i_p)} A^{(i_1)}} \cdots \mat{A^{(i_{p-1})}} \right)
                                                        = x_{i_p,i_1,i_2,\cdots, i_{p-1}} \; . \\
\end{eqnarray*}

Let us now suppose that the vector $\vec x$ has the bit-shift symmetry and let
$$ x_{i_1,i_2,\dots,i_p} = \trace \left( \mat{ B_1^{(i_1)} B_2^{(i_2)} \cdots B_p^{(i_p)} } \right)$$
be any \mps representation for $\vec x$.
Then the construction
\begin{equation}\label{eq:BitShiftSymmetryConstructMPS}
\mat{A^{(i_j)}} = \frac{1}{\sqrt[p]{p}}
\pmat{ \mat{0} & \mat{B_1^{(i_j)}} & \mat{0} &  &  \\
               & \mat{0} & \mat{B_2^{(i_j)}} & \ddots & \\
               &         &      \ddots       & \ddots & \mat{0}\\
               &         &                   &  \mat{0} & \mat{B_{p-1}^{(i_j)}} \\
\mat{B_p^{(i_j)}} &      &                   &          & \mat{0}
}
\end{equation}
leads to a site-independent representation of $\vec x$.
\end{proof}
\begin{remark}
The construction~\myref{eq:BitShiftSymmetryConstructMPS}
introduces an augmentation of the matrix size by the factor $p$.
\end{remark}

The bit-shift symmetry can also be generalized to block-shift symmetry.
Assume that a block of $r$ \MPS matrix pairs is repeated, i.e.
$$ \left( \mat{A_1^{(i_1)} A_2^{(i_2)}} \cdots \mat{A_r^{(i_r)}} \right)
\left( \mat{A_1^{(i_{r+1})} A_2^{(i_{r+2})}} \cdots \mat{A_r^{(i_{2r})}} \right)
\cdots \left( \mat{A_1^{(i_{p-r+1})} A_2^{(i_{p-r+2})}} \cdots \mat{A_r^{(i_{p})}} \right) \; $$
to obtain symmetries of the form
$$x_{i_1,\dots,i_r;i_{r+1},\dots,i_{2r}; \dots ; i_{p-r+1},\dots,i_p} =
x_{i_{r+1},\dots,i_{2r};\dots; i_{p-r+1},\dots,i_p;i_1,\dots,i_r} \; .$$

\subsubsection*{Normal Form for the Bit-Shift Symmetry}
In the above periodic TI ansatz \myref{eq:defTIMPS} we can
replace each $\mat A$ by $\mat{MA}\mat M^{-1}$ with a nonsingular $\mat M$ resulting in the same
vector $\vec{x}$. Using the Schur normal form $\mat{A^{(0)}}=\mat Q^{\herm}\mat{R^{(0)}Q}$ or the
Jordan canonical form $\mat{A^{(0)}}=\mat S^{-1}\mat{J_{A}^{(0)}S}$
we propose to normalize the
\MPS form by replacing the matrix pair $(\mat{A^{(0)}},\mat{A^{(1)}})$ by
$$
\bigl(\mat{\tilde A^{(0)}},\mat{\tilde A^{(1)}}\bigr)=\bigl(\mat{R^{(0)}},\mat{QA^{(1)}} \mat Q^{\herm}\bigr)
    \quad \text{or} \quad
    \bigl(\mat{\tilde A^{(0)}},\mat{\tilde A^{(1)}}\bigr)=\bigl(\mat{ J_A^{(0)}},\mat{S} \mat{A^{(1)}} \mat S^{-1}\bigr)
$$
resulting in a more
compact representation of $\vec{x}$ with less free parameters. For Hermitian
$\mat{A^{(0)}}$ and $\mat{A^{(1)}}$ the eigenvalue decomposition of $\mat{A^{(0)}}=\mat Q^{\herm}\mat{D^{(0)}Q}$
with diagonal matrix $\mat{D^{(0)}}$ can be used in the same way leading to
the normal form
$$(\mat{\tilde A^{(0)}},\mat{\tilde A^{(1)}})=(\mat{D^{(0)}},\mat{Q} \mat{A^{(1)}} \mat Q^{\herm})$$
with $\mat{\tilde A^{(0)}}$ as real diagonal matrix and $\mat{\tilde A^{(1)}}$ as Hermitian matrix.

\subsubsection{Reverse Symmetry}
In this subsection we consider the \defini{reverse symmetry}
\begin{equation}\label{eq:reverseSymmetry}
x_{i_1,\dots,i_p} = \bar x_{i_p,\dots,i_1} \; .
\end{equation}
The following theorem shows a direct connection between the reverse symmetry and
an \MPS representation with the special symmetry relations
\begin{equation}\label{eq:MatricesReverseSymmetry}
\bigl( \mat{A_j^{(i_j)}} \bigr)^{\herm} = \mat{S_{p-j}^{\inv}} \mat{A_{p+1-j}^{(i_j)} S_{p+1-j}} \quad \text{ for all } j=1,\dots,p
\end{equation}
with regular matrices $\mat{S_j}$ of appropriate size,
which additionally fulfill the consistency conditions
\begin{equation}\label{eq:reverseSymmetryRelationsInS}
\mat{S_0} = \mat{S_p} \qquad \text{and} \qquad \mat{S_j^{\herm}} = \mat{S_{p-j}} \ \text{ for } j=1,\dots,p\; .
\end{equation}

\begin{theorem}\label{thm:reverseSymmetryRelations}
If the \MPS matrices fulfill the symmetry relations \myref{eq:MatricesReverseSymmetry},
the vector to be represented has the reverse symmetry property \myref{eq:reverseSymmetry}.
Vice versa, for any vector $\vec x$ fulfilling the reverse symmetry,
we may state an \MPS representation for $\vec x$ fulfilling the relations~\myref{eq:MatricesReverseSymmetry}.
\end{theorem}
\begin{proof}
For the vector $\vec x$ to be represented, the relations \myref{eq:MatricesReverseSymmetry}
lead to
{\allowdisplaybreaks
\begin{align*}
x_{i_1,\dots,i_p} & = \trace \bigl( \mat{A_1^{(i_1)}} \mat{A_2^{(i_2)}} \cdots \mat{A_p^{(i_p)}}   \bigr) \\
& = \overline{ \trace \left( \mat{A_1^{(i_1)}} \mat{A_2^{(i_2)}} \cdots \mat{A_p^{(i_p)}} \right)^{\herm} } \\
& = \overline{ \trace \left( \mat{A_p^{(i_p) \herm}} \mat{A_{p-1}^{(i_{p-1}) \herm }} \cdots
                \mat{A_2^{(i_2) \herm}} \mat{A_1^{(i_1) \herm}} \right) } \\
& = \overline {\trace \left( \bigl( \mat{S_p^{\inv}} \mat{A_1^{(i_p)}} \mat{S_1} \bigr)
                    \bigl( \mat{S_1^{\inv}} \mat{A_2^{(i_{p-1})}} \mat{S_2} \bigr)
                    \cdots
                    \bigl( \mat{S_{p-1}^{\inv}} \mat{A_p^{(i_1)}} \mat{S_p} \bigr)
            \right) } \\
& = \overline{ \trace \left( \mat{A_1^{(i_p)}} \mat{A_2^{(i_{p-1})}}  \cdots \mat{A_p^{(i_1)}} \right) } \\
& = \bar x_{i_p,\dots,i_1} \; ,
\end{align*}}
a reverse symmetric vector.

So far we have seen that the relations \myref{eq:MatricesReverseSymmetry} lead to the representation
of a vector having the reverse symmetry.
Contrariwise, it is possible to indicate an \mps representation fulfilling
the relations \myref{eq:MatricesReverseSymmetry} for any reverse symmetric vector.
To see this we consider any \MPS for $\vec x$:
\begin{equation*}
x_{i_1,i_2,\dots,i_p} = \trace \bigl( \mat{B_1^{(i_1)}} \mat{B_2^{(i_2)}} \cdots \mat{B_p^{(i_p)}}   \bigr)
\end{equation*}
with matrices $\mat{B_j^{(i_j)}}$ of size $D_j \times D_{j+1}$.
Such an \mps representation always exists, compare Lemma~\ref{lemma:ExistenceUniqueMPSGauge}.

Let us start with the case where this \mps representation is in PBC form.
The reverse symmetry $x_{i_1,i_2,\dots,i_p} = \bar x_{i_p,i_{p-1},\dots,i_1}$ leads to
\begin{equation}\label{eq:ReverseSymmetryConstructMPSWithRelations}
\begin{split}
x_{i_1,i_2,\dots,i_p} & = \tfrac{1}{2} \left( x_{i_1,i_2,\dots,i_p} + \bar x_{i_p,i_{p-1},\dots,i_1} \right) \\
& = \tfrac{1}{2} \left( \trace \bigl( \mat{B_1^{(i_1)}} \mat{B_2^{(i_2)}} \cdots \mat{B_p^{(i_p)}}   \bigr) +
                        \overline{ \trace \bigl( \mat{B_1^{(i_p)}} \mat{B_2^{(i_{p-1})}} \cdots \mat{B_p^{(i_1)}}   \bigr) }
                \right) \\
& = \tfrac{1}{2} \left( \trace \bigl( \mat{B_1^{(i_1)}} \mat{B_2^{(i_2)}} \cdots \mat{B_p^{(i_p)}}   \bigr) +
                        \trace \bigl( \mat{B_p^{(i_1) \herm }} \mat{ B_{p-1}^{(i_2) \herm }} \cdots \mat{B_1^{(i_p) \herm}}   \bigr)
                \right) \\
& = \tfrac{1}{2} \trace \left(
                          \begin{array}{cc}
                            \mat{B_1^{(i_1)}} \mat{B_2^{(i_2)}} \cdots \mat{B_p^{(i_p)}} & \mat 0 \\
                            \mat 0 & \mat{B_p^{(i_1) \herm }} \mat{ B_{p-1}^{(i_2) \herm }} \cdots \mat{B_1^{(i_p) \herm}}
                          \end{array}
                        \right) \\
& = \tfrac{1}{2} \trace \left( \pmat{ \mat{B_1^{(i_1)}} & \mat 0 \\ \mat 0 & \mat{B_p^{(i_1) \herm}}}
                  \pmat{ \mat{B_2^{(i_2)}} & \mat 0 \\ \mat 0 & \mat{B_{p-1}^{(i_2) \herm}}} \cdots
                  \pmat{ \mat{B_p^{(i_p)}} & \mat 0 \\ \mat 0 & \mat{B_1^{(i_p) \herm}}}
            \right) \; .
\end{split}
\end{equation}
We may now define
\begin{equation}
\mat{A_j^{(i_j)}} := \tfrac{1}{\sqrt[p]{2}} \pmat{ \mat{B_j^{(i_j)}} & \mat 0 \\ \mat 0 & \mat{B_{p+1-j}^{(i_j) \herm}} }
\end{equation}
and obtain
{\allowdisplaybreaks
\begin{align*}
\mat{A_j^{(i_j) \herm}} = \tfrac{1}{\sqrt[p]{2}} \pmat{ \mat{B_j^{(i_j) \herm}} & \mat 0 \\ \mat 0 & \mat{B_{p+1-j}^{(i_j)}} }
& = \tfrac{1}{\sqrt[p]{2}} \pmat{ \mat 0 \ & \ \mat I \\ \mat I \ & \ \mat 0}
    \pmat{ \mat{B_{p+1-j}^{(i_j)}} & \mat 0 \\ \mat 0 & \mat{B_{j}^{(i_j) \herm}} }
    \pmat{ \mat 0 \ & \ \mat I \\ \mat I \ & \ \mat 0} \\
& = \pmat{ \mat 0 \ & \ \mat I \\ \mat I \ & \ \mat 0}
    \mat{ A_{p+1-j}^{(i_j)}}
    \pmat{ \mat 0 \ & \ \mat I \\ \mat I \ & \ \mat 0}
\end{align*}}
with $\mat I$ being identities of appropriate size. Hence, the choice
\begin{equation}\label{eq:proofDefinitionS_j}
\mat{S_j} = \pmat{ \mat 0 \ & \ \mat{I_{D_{j+1} } } \\ \mat{I_{D_{p+1-j}}} \ & \ \mat 0} \quad \text{ for all } j=1,\dots,p
\end{equation}
gives $\mat{A_j^{(i_j) \herm}} = \mat{S_{p-j}^{\inv}} \mat{A_{p+1-j}^{(i_j)} S_{p+1-j}}$,
the desired matrix relations \myref{eq:MatricesReverseSymmetry}.

In the OBC case we can proceed in a similar way, but at both ends $j=1$ and $j=p$ something special happens:
as we want to preserve the OBC character of the \mps representation, the matrices
$\mat{A_1^{(i_1)}}$ and $\mat{A_p^{(i_p)}}$ have to be vectors as well.
Therefore we define
$$ \mat{A_1^{(i_1)}} = \tfrac{1}{\sqrt[p]{2}} \pmat{ \mat{B_1^{(i_1)}} \ & \ \mat{B_p^{(i_1) \herm}} }
\quad \text{and} \quad
\mat{A_p^{(i_p)}} = \tfrac{1}{\sqrt[p]{2}} \pmat{ \mat{B_p^{(i_p)}} \\ \mat{B_1^{(i_p) \herm}} } \; .
$$
The choice $\mat{S_p}=1$ leads to the  desired relation
$\bigl( \mat{A_1^{(i_1)}} \bigr)^{\herm} = \mat{S_{p-1}^{\inv}} \mat{A_{p}^{(i_1)} S_{p}}$.
\end{proof}
\begin{remark}
\begin{enumerate}
\item The proof shows that the reverse symmetry can occur in the periodic boundary case, but also
for the open boundary case where $\mat{A_1^{(i_1)}}$ and $\mat{A_p^{(i_p)}}$ specialize to vectors.
Then, $\mat{S_0} = \mat{S_p} \stackrel{\myref{eq:reverseSymmetryRelationsInS}}{=} \mat{S_0^{\herm}}$
are simply (even real) scalars.
\item In the proof, the matrices $\mat{S_j}$ can be chosen to be unitary, compare \myref{eq:proofDefinitionS_j}.
Thus, they can be diagonalized by a unitary transform $\mat{V_j}$ giving
$$ \mat{S_j} = \mat{V_j \Delta_j V_j^{\herm}}  \qquad \text{with diagonal and unitary } \mat{\Delta_j}  \; .$$
Because of
$\mat{S_{p-j}^{\inv}} = \mat{S_{p-j}^{\herm}} \stackrel{\myref{eq:reverseSymmetryRelationsInS}}{=} \mat{S_j}$
the relations \myref{eq:MatricesReverseSymmetry} read
\begin{align*}
\bigl( \mat{A_j^{(i_j)} \bigr)^{\herm}} & = \mat{S_j A_{p+1-j}^{(i_j)} S_{p+1-j}} \\
    & =  \bigl( \mat{V_j \Delta_j V_j^{\herm}} \bigr) \mat{A_{p+1-j}^{(i_j)}}
            \bigl( \mat{V_{p+1-j} \Delta_{p+1-j} V_{p+1-j}^{\herm}} \bigr) \; .
\end{align*}
The last equation can be rewritten to
$$ \mat{V_j^{\herm}} \bigl( \mat{A_j^{(i_j)} \bigr)^{\herm}} \mat{V_{p+1-j}} =
\mat{ \Delta_j V_j^{\herm} A_{p+1-j}^{(i_j)}} \mat{V_{p+1-j} \Delta_{p+1-j} } \; . $$
Defining $\mat{\tilde A_j^{(i_j)}} := \mat{ V_{p+1-j}^{\herm} A_j{(i_j)} V_j }$, the relations \myref{eq:MatricesReverseSymmetry} take the form
\begin{equation}\label{eq:reverseSymmetryHermitianSRelations}
 \bigl( \mat{\tilde A_j^{(i_j)}} \bigr)^{\herm} = \mat{ \Delta_j \tilde A_{p+1-j}^{(i_j)} \Delta_{p+1-j}}
 \end{equation}
with unitary diagonal matrices $\mat{\Delta_j}$ fulfilling $\mat{\Delta_j^{\herm}}=\mat{\Delta_{p-j}}$.
If the matrices $\mat{S_j}$ are unitary and also Hermitian, the diagonal matrices $\mat{\Delta_j}$
have values $\pm 1$ on the main diagonal.
\item In the PBC case with \mps matrices $\mat{B_j^{(i_j)}}$ of equal size $D \times D$,
the $\mat{S_j}$ matrices~\myref{eq:proofDefinitionS_j}
can be chosen to be site-independent, Hermitian, and unitary.
In this case the relations~\myref{eq:reverseSymmetryHermitianSRelations}
are fulfilled by $\mat{\Delta_j} = \operatorname{diag}(\mat{I_D},-\mat{I_D})$.
\end{enumerate}
\end{remark}

\subsubsection*{Normal Form for the Reverse Symmetry}
In the following theorem we propose a normal form for \MPS
representations of reverse symmetric vectors.
\begin{theorem}
Let $\vec x \in \mathbb C^{2^p}$ be a vector with the reverse symmetry.
If $p = 2 m$ is even, $\vec{x}$ can be represented by an \mps of the form
\begin{equation}\label{eq:ReverseSymmetryNormalFormEvenp}
x_{i_1,i_2,\dots,i_p} =
\trace \left( \bigl( \mat{U_1^{(i_1)}} \cdots \mat{U_{m}^{(i_{m})}} \bigr) \mat{\Sigma}
    \bigl( \mat{U_{m}^{(i_{m+1})\herm}} \cdots \mat{U_1^{(i_p) \herm}} \bigr) \mat{\Lambda} \right)
\end{equation}
and if $p = 2 m + 1$ is odd, the representation reads
\begin{equation}\label{eq:ReverseSymmetryNormalFormOddp}
x_{i_1,i_2,\dots,i_p} =
\trace \left( \bigl( \mat{U_1^{(i_1)}} \cdots \mat{U_{m}^{(i_{m})}} \bigr) \mat{U_{m+1}^{(i_{m+1})}} \mat{\Sigma}
    \bigl( \mat{U_{m}^{(i_{m+2})\herm}} \cdots \mat{U_1^{(i_p) \herm}} \bigr) \mat{\Lambda} \right)
\end{equation}
with unitary matrices  $\mat{U_j^{(i_j)}}$ and real and diagonal matrices $\mat{\Sigma}$ and $\mat{\Lambda}$.
\end{theorem}
\begin{proof}
We start with an \mps of the form~\myref{eq:MatricesReverseSymmetry} to represent the given vector~$\vec x$,
compare Theorem~\ref{thm:reverseSymmetryRelations}.

In the case of $p=2m$ being even, we obtain
{\allowdisplaybreaks
\begin{align}\nonumber
x_{i_1,\dots,i_p} & = \trace \left( \mat{A_1^{(i_1)}} \mat{A_2^{(i_2)}} \cdots
                        \mat{A_m^{(i_m)}} \mat{A_{m+1}^{(i_{m+1})}}  \cdots \mat{A_p^{(i_p)}} \right) \\ \nonumber
            & \stackrel{\myref{eq:MatricesReverseSymmetry}}{=} \trace \left( \mat{A_1^{(i_1)}} \mat{A_2^{(i_2)}} \cdots
                        \mat{A_m^{(i_m)}}
                        \bigl( \mat{S_m^{\herm}} \mat{A_{m}^{(i_{m+1}) \herm}} \mat{S_{m-1}^{-\herm}} \bigr)
                        \cdots \bigl( \mat{S_1^{\herm}} \mat{A_1^{(i_p) \herm}} \mat{S_p^{- \herm}} \bigr) \right) \\ \label{eq:ReverseSymmetryNormalformEven}
            & = \trace \left( \bigl( \mat{A_1^{(i_1)}} \mat{A_2^{(i_2)}} \cdots
                        \mat{A_m^{(i_m)}} \bigr)
                        \mat{S_m^{\herm}}
                        \bigl( \mat{A_{m}^{(i_{m+1}) \herm}} \cdots \mat{A_1^{(i_p) \herm}} \bigr)
                                \mat{S_p^{- \herm}} \right) \; .
\end{align}}
Following \myref{eq:reverseSymmetryRelationsInS} we obtain $\mat{S_p^{\herm}} = \mat{S_p}$
and thus we may factorize $\mat{S_p^{\text{$-\mathrm H$}}} = \mat{W \Lambda W^{\herm}}$.
Using \myref{eq:traceInvariance}, Eqn.~\myref{eq:ReverseSymmetryNormalformEven} reads
$$
x_{i_1,\dots,i_p} =
\trace \left( \bigl( \mat{W^{\herm} A_1^{(i_1)}} \bigr) \mat{A_2^{(i_2)}} \cdots
                        \mat{A_m^{(i_m)}}
                        \mat{S_m^{\herm}}
                        \mat{A_{m}^{(i_{m+1}) \herm}} \cdots \mat{A_2^{(i_{p-1}) \herm}}
                        \bigl( \mat{W^{\herm} A_1^{(i_p)}} \bigr)^{\herm}
                                \mat{\Lambda} \right) \; .
$$
We use the SVD of $\mat{W^{\herm} A_1}$,
\begin{equation*}
\pmat{ \mat{W^{\herm} A_1^{(0)}} \\ \mat{ W^{\herm} A_1^{(1)}} } =
    \pmat{ \mat{U_1^{(0)}} \\ \mat{U_1^{(1)}} } \mat{ \Lambda_1 } \mat{ V_1} \; ,
\end{equation*}
to replace $\mat{W^{\herm} A_1}$ at both ends. We then obtain
\begin{equation*}
\trace \left( \mat{ U_1^{(i_1)}} \bigl( \mat{\Lambda_1 V_1} \mat{A_2^{(i_2)}} \bigr) \cdots
                        \mat{A_m^{(i_m)}}
                        \mat{S_m^{\herm}}
                        \bigl( \mat{A_{m}^{(i_{m+1}) \herm}} \cdots
                        \bigl( \mat{A_2^{(i_{p-1}) \herm}}
                            \mat{ V_1^{\herm} \Lambda_1} \mat{U_1^{(i_p) \herm}} \bigr)
                        \mat{\Lambda} \right) \; .
\end{equation*}
We proceed with the SVD for $\mat{\Lambda_1 V_1} \mat{A_2^{(i_2)}}$, i.e.
$\mat{\Lambda_1 V_1} \mat{A_2^{(i_2)}} = \mat{U_2^{(i_2)} \Lambda_2 V_2}$, to obtain
\begin{equation*}
\trace \left( \mat{U_1^{(i_1)}} \mat{U_2^{(i_2)}} \bigl( \mat{\Lambda_2 V_2} \mat{A_3^{(i_3)}} \bigr) \cdots
                        \mat{A_m^{(i_m)}}
                        \mat{S_m^{\herm}}
                        \mat{A_{m}^{(i_{m+1}) \herm}} \cdots
                         \bigl( \mat{A_3^{(i_3) \herm}} \mat{ V_2^{\herm} \Lambda_2} \bigr) \mat{U_2^{(i_{p-1}) \herm}} \mat{U_1^{(i_p) \herm}} \mat{\Lambda} \right) \; .
\end{equation*}
Proceeding in an iterative way finally gives
\begin{equation*}
\trace \biggl( \mat{ \bigl( U_1^{(i_1)}} \mat{U_2^{(i_2)}} \cdots \mat{U_m^{(i_m)}} \bigr)
\underbrace{\bigl( \mat{\Lambda_m V_m} \mat{S_m^{\herm}} \mat{V_m^{\herm}} \mat{\Lambda_m} \bigr)}_{=:\mat{C}}
                 \bigl( \mat{U_{m}^{(i_{m+1}) \herm}} \cdots
                         \mat{U_2^{(i_{p-1}) \herm}} \mat{U_1^{(i_p) \herm}} \biggr)
                                \mat{\Lambda} \Bigr) \; .
\end{equation*}
For $j=m$, Eqn. \ref{eq:reverseSymmetryRelationsInS} yields $\mat{S_m^{\herm}} = \mat{S_m}$
and thus $\mat C$ is also Hermitian leading to
$$ \mat{\Lambda_m V_m} \mat{S_m^{\herm}} \mat{V_m^{\herm}} \mat{\Lambda_m} =
\mat{C} = \mat{C^{\herm}} = \mat{X \Sigma X^{\herm}} \; .$$
with unitary $\mat{X}$ and real diagonal $\mat{\Sigma}$.
Altogether we obtain the \mps representation
\begin{equation*}
\trace \left( \mat{U_1^{(i_1)}} \mat{U_2^{(i_2)}} \cdots \bigl( \mat{U_m^{(i_m)} X} \bigr)
                \mat{\Sigma}
                 \bigl( \mat{X^{\herm}} \mat{U_{m}^{(i_{m+1}) \herm}} \bigr) \cdots
                         \mat{U_2^{(i_{p-1}) \herm}} \mat{U_1^{(i_p) \herm}}
                                 \mat{\Lambda} \right) \; .
\end{equation*}
Replacing $\mat{U_m^{(i)}}$ by the unitary matrix $\mat{U_m^{(i)} X}$ gives
the desired normal form~\myref{eq:ReverseSymmetryNormalFormEvenp}.


For the odd case $p=2m+1$, we may proceed in a similar way
and replace all factors up to the interior one related to $j=m+1$ by unitary matrices to obtain
\begin{equation*}
\trace \biggl( \bigl( \mat{U_1^{(i_1)}} \mat{U_2^{(i_2)}} \cdots \mat{U_m^{(i_m)}} \bigr)
                 \underbrace{\bigl( \mat{\Lambda_m V_m} \mat{A_{m+1}^{(i_{m+1})}} \mat{S_m^{\herm}} \mat{V_m^{\herm}} \mat{\Lambda_m} \bigr)}_{=:\mat{C^{(i_{m+1})}}}
                 \bigl( \mat{U_{m}^{(i_{m+2}) \herm}} \cdots \mat{U_1^{(i_p) \herm}} \bigr)
                                \mat{\Lambda} \biggr) \; .
\end{equation*}
For site $j=m+1$ the supposed matrix relations lead to
\begin{equation*}
\left( \mat{A_{m+1}^{(i_{m+1})} S_m^{\herm}} \right)^{\herm}
= \mat{S_m} \left( \mat{A_{m+1}^{(i_{m+1})}} \right)^{\herm}
\stackrel{\myref{eq:MatricesReverseSymmetry}}{=}
\mat{S_{m}} \bigl( \mat{S_{m}^{\text{-1}} A_{m+1}^{(i_{m+1})} S_{m+1} } \bigr)
\stackrel{\myref{eq:reverseSymmetryRelationsInS}}{=} \mat{A_{m+1}^{(i_{m+1})} S_m^{\herm} }
\end{equation*}
and thus the matrices $\mat{C^{(i_{m+1})}}$ are both Hermitian.
Using the SVD gives
\begin{equation}\label{eq:ReverseSymmetryProofNormalFormSVDOdd}
\mat{C^{(i_{m+1})}} = \mat{U_{m+1}^{(i_{m+1})} \Sigma X } =
\mat{ X^{\herm} \Sigma U_{m+1}^{(i_{m+1}) \herm} } \; .
\end{equation}
Hence, for the overall representation we obtain
\begin{align*}
& \trace \left( \bigl( \mat{U_1^{(i_1)}} \cdots \mat{U_m^{(i_m)}} \bigr)
                 \bigl( \mat{U_{m+1}^{(i_{m+1})} \Sigma X } \bigr)
                 \bigl( \mat{U_{m}^{(i_{m+2}) \herm}} \cdots
                         \mat{U_1^{(i_p) \herm}} \bigr) \mat{\Lambda} \right) \\
 = & \trace \left( \mat{U_1^{(i_1)}} \cdots
                    \bigl( \mat{U_m^{(i_m)}} \mat{X^{\herm}} \bigr)
                 \bigl( \mat{X U_{m+1}^{(i_{m+1})} } \bigr) \mat{\Sigma}
                 \bigl( \mat{ U_{m}^{(i_{m+2})}} \mat{X^{\herm}} \bigr)^{\herm} \cdots
                         \bigl( \mat{U_1^{(i_p)}} \bigr)^{\herm}
                                \mat{\Lambda} \right) \; .
\end{align*}
Replacing $\mat{U_m^{(i)}}$ by $\mat{U_m^{(i)} X^{\herm}}$ and
$\mat{U_{m+1}^{(i)}}$ by $\mat{X U_{m+1}^{(i)} }$
leads to the normal form~\myref{eq:ReverseSymmetryNormalFormOddp} for the odd case.
\end{proof}
\begin{remark}
In the odd case we may also use the right-side SVD factorization
\mbox{$\mat{C^{(i_{m+1})}} = \mat{ X^{\herm} \Sigma U_{m+1}^{(i_{m+1}) \herm} }$}
in Eqn.~\ref{eq:ReverseSymmetryProofNormalFormSVDOdd}
leading to the normal form
$$\trace \left( \bigl( \mat{U_1^{(i_1)}} \cdots \mat{U_{m}^{(i_{m})}} \bigr)
        \mat{\Sigma} \mat{U_{m+1}^{(i_{m+1}) \herm}}
        \bigl( \mat{U_{m}^{(i_{m+2})\herm}} \cdots \mat{U_1^{(i_p) \herm}} \bigr) \mat{\Lambda} \right) \; .
$$
This ambiguity is reasonable as the interior factor in the odd case only has itself as counter part:
$\mat{A_1} \leftrightarrow \mat{A_p}$, $\mat{A_2} \leftrightarrow \mat{A_{p-1}}$, $\dots$ ,
$\mat{A_{m}} \leftrightarrow \mat{A_{m+2}}$, $\mat{A_{m+1}} \leftrightarrow \mat{A_{m+1}}$.
\end{remark}

\subsubsection*{Reverse Symmetry in TI Representations}
Let us finally consider the reverse symmetry in TI representations.
This additional property allows us to use site-independent matrices $\mat{S_j} = \mat{S}$,
which are  Hermitian, compare \myref{eq:reverseSymmetryRelationsInS}.
Then the relations \myref{eq:MatricesReverseSymmetry} take the form
$$\bigl( \mat{A^{(i)}} \bigr)^{\herm} = \mat{S^{\inv} A^{(i)} \mat{S}} \Longleftrightarrow
\bigl( \mat{A^{(i)} S } \bigr)^{\herm} = \mat{ A^{(i)} \mat{S}} \; .
$$
Thus, we can represent the vector with Hermitian matrices $\mat{\tilde A^{(i)}} := \mat{A^{(i)} S }$.
In the QI society one can find considerations on TI systems using real symmetric matrices, compare \cite{Pirvu11ExploitingTI}.

\subsubsection{Bit-Flip Symmetry}
Here we focus on the representation of symmetric and skew-symmetric vectors
appearing, e.g., as eigenvectors of symmetric persymmetric matrices
(see Lemma~\ref{lemma:EigenvectorsPersymmMatrices}).
We will use the \defini{bit-flip operator} $\bar i := 1-i$ for $i \in\{0,1\}$.
First we show that the symmetry condition $\mat J \vec{x} = \vec{x}$
corresponds to the \defini{bit-flip symmetry}
$$ x_{i_1,i_2,\dots,i_p} = x_{\bar{i}_1,\bar{i}_2,\dots,\bar{i}_p} \; .$$
To see this we consider
{\allowdisplaybreaks
\begin{align*}
 \mat J \vec{x} & = (\mat{J_2} \otimes \cdots \otimes \mat{J_2})
                        \Biggl( \sum\limits_{i_1,\dots,i_p} x_{i_1,i_2,\dots,i_p}
                                \bigl( \vec{e_{i_1}} \otimes \cdots \otimes \vec{e_{i_p}} \bigr) \Biggr ) \\
        & = \sum\limits_{i_1,\dots,i_p} x_{i_1,i_2,\dots,i_p} \bigl( (\mat{J_2} \vec{e_{i_1}} ) \otimes \cdots
        \otimes (\mat{J_2} \vec{e_{i_p}} ) \bigr) \\
        & = \sum\limits_{i_1,\dots,i_p} x_{i_1,i_2,\dots,i_p}
                    \bigl( \vec{e_{\overline{i}_1}} \otimes \cdots \otimes \vec{e_{\bar{i}_p}} \bigr) \\
        & = \sum\limits_{i_1,\dots,i_p} x_{\bar{i}_1,\bar{i}_2,\dots,\bar{i}_p}
                    \bigl( \vec{e_{{i_1}}} \otimes \cdots \otimes \vec{e_{{i_p}}} \bigr) \; . 
 \end{align*}}
Hence we obtain
$$ \mat{J} \vec{x} = \vec{x} \ \Longleftrightarrow \ x_{i_1,i_2,\dots,i_p} =
x_{\bar{i}_1,\bar{i}_2,\dots,\bar{i}_p}
\text{ for all } i_1,\dots,i_p=0,1 \; .$$
Analogously, for a skew-symmetric vector $\vec{x}$ one gets
$x_{i_1,i_2,\dots,i_p} = - x_{\bar{i}_1,\bar{i}_2,\dots,\bar{i}_p}$.

In order to express these relations in the \MPS formalism consider 
\begin{equation}\label{eq:defMPSPairsInvolution}
\mat{A_j^{(1)}} = \mat{U_j A_j^{(0)} U_{j+1 \ {\rm mod} \ p}} \quad \text{ for } \ j=1,\dots,p
\end{equation}
with $\mat{U_j}$ being involutions, i.e. $\mat U_{\mat j}^2 = \mat I$ (\cite{Sandvik07Variational}).
Then Eqn. \myref{eq:defMPSPairsInvolution} can also be expressed vice versa to give
\begin{equation}\label{eq:MPSRelationInvolution}
\mat{A_j^{(i_j)}} = \mat{U_j A_j^{(\bar{i}_j)} U_{j+1 \ {\rm mod} \ p}} \; .
\end{equation}
The following lemma shows the correspondence between these relations
and the bit-flip symmetry.
\begin{theorem}
If the matrix pairs $(\mat{A_j^{(0)}},\mat{A_j^{(1)}})$ are connected via involutions
as in \myref{eq:defMPSPairsInvolution}
the represented vector has the bit-flip symmetry and is hence symmetric.
Contrariwise any symmetric vector can be represented by an \mps
fulfilling condition~\myref{eq:defMPSPairsInvolution}.
\end{theorem}
\begin{proof}
The matrix relations~\myref{eq:defMPSPairsInvolution} translate into the symmetry of the represented vector
\begin{eqnarray*}
x_{i_1,i_2,\dots,i_p} & = & \trace \left(  \mat{A_1^{(i_1)}} \cdot \mat{A_2^{(i_2)}} \cdots \mat{A_p^{(i_p)}} \right) \\
& \stackrel{\myref{eq:MPSRelationInvolution}}{=} & \trace \left(  \left( \mat{U_1 A_1^{( \bar{i}_1 )} U_2} \right)  \cdot
        \left( \mat{U_2 A_2^{(\bar{i}_2)} U_3} \right) \cdots \left( \mat{U_p A_p^{(\bar{i}_p)} U_1} \right) \right) \\
& = & \trace \left(  \mat{A_1^{( \bar{i}_1 )} A_2^{(\bar{i}_2)}} \cdots  \mat{A_p^{(\bar{i}_p)}} \right) \\
& = & x_{\bar{i_1},\bar{i_2},\dots,\bar{i_p}} \; .
\end{eqnarray*}

Let us now consider the construction of an \mps representation~\myref{eq:defMPSPairsInvolution}
for a symmetric vector $\vec x$ fulfilling the bit-flip symmetry
$x_{i_1,i_2,\dots,i_p} = x_{\bar{i_1},\bar{i_2},\dots,\bar{i_p}}$.
To this end we start with any \mps representation
$$ x_{i_1,i_2,\dots,i_p} = \trace \left( \mat{B_1^{(i_1)}} \mat{B_2^{(i_2)}} \cdots \mat{B_p^{(i_p)}} \right) $$
with $D_j \times D_{j+1}$ matrices $\mat{B_j^{(i_j)}}$.
Starting from the identity
$$
x_{i_1,i_2,\dots,i_p} = \frac{1}{2} \left( x_{i_1,i_2,\dots,i_p} + x_{\bar{i_1},\bar{i_2},\dots,\bar{i_p}} \right)
$$
we may proceed in a similar way as in \myref{eq:ReverseSymmetryConstructMPSWithRelations}
for the reverse symmetry to obtain
$$ x_{i_1,i_2,\dots,i_p} =
\tfrac{1}{2} \trace \left( \pmat{ \mat{B_1^{(i_1)}} & \mat 0 \\ \mat 0 & \mat{B_1^{(\bar i_1)}}}
                  \pmat{ \mat{B_2^{(i_2)}} & \mat 0 \\ \mat 0 & \mat{B_{2}^{(\bar i_2)}}} \cdots
                  \pmat{ \mat{B_p^{(i_p)}} & \mat 0 \\ \mat 0 & \mat{B_p^{(\bar i_p)}}}
                    \right) \; . $$
This equation motivates the definition
$$ \mat{A_j^{(i_j)}} := \pmat{ \mat{B_j^{(i_j)}} & \mat 0 \\ \mat 0 & \mat{B_j^{(\bar i_j)}}} \; .$$
In the OBC case the first and last matrices have to specialize to vectors:
$$ \mat{A_1^{(i_1)}} = \pmat{ \mat{B_1^{(i_1)}} \ & \ \mat{B_1^{(\bar i_1)}}}
\qquad \text{and} \qquad \mat{A_p^{(i_p)}} = \pmat{ \mat{B_p^{(i_p)}} \\ \mat{B_p^{(\bar i_p)}}} \; .
$$
Using the involutions
\begin{equation}\label{eq:BitFlipSymmetryProofConstructInvolution}
\mat{U_j} := \pmat{ \mat{0} & \mat{I_{D_j}} \\ \mat{I_{D_j}} & \mat{0} } \; \text{ for } j=1,\dots,p
\end{equation}
gives the desired relations~\myref{eq:defMPSPairsInvolution}.
In the OBC case we have to define $\mat{U_1} = 1$.
\end{proof}

\begin{remark}
If we want to represent a skew-symmetric vector $\vec x = - \mat J \vec x$, we may also use the relations
\myref{eq:MPSRelationInvolution} at all sites up to one, say site $1$, where we would have to add a negative sign:
$\mat{A_1^{(i_1)}} = - \mat{U_1 A_1^{(\bar{i}_1)} U_2}$.
However, in the special TI case, where all matrix pairs have to be identical, this is not possible:
the relations~\myref{eq:MPSRelationInvolution} would read
\begin{equation}\label{eq:TIandInvolutions}
\pmat{ \mat{A_j^{(0)}} \ & \ \mat{A_j^{(1)}}} = \pmat{ \mat{A} \ & \ \mat{U A V}}
\end{equation}
at every site $j$ with site-independent involutions $\mat U$ and $\mat V$.
Therefore, in the periodic \TIMPS ansatz \myref{eq:TIandInvolutions}
applied to symmetric-persymmetric Hamiltonians, only symmetric eigenvectors can occur.
\end{remark}

\subsubsection*{Normal Form for the Bit-Flip Symmetry}

As every involution, $\mat{U_j}$ may only have eigenvalues $\in \{-1,1\}$ and thus 
\begin{equation}\label{eq:diagonalizeInvolution}
 \mat{U_j} = \mat{S_j^{\text{$-1$}} D_{j;\pm 1} S_j} \; ,
\end{equation}
where $\mat{D_{j;\pm 1}}$ is a diagonal matrix with entries $\pm 1$:
the Jordan canonical form implies
$\mat{U_j}=\mat{S_j^{\text{$-1$}}J_{U_j}S_j}$.
Moreover, the Jordan blocks in
$\mat{J_{U_j}}$ have to be involutions as well,
so $\mat{J_{U_j}^{\text{$2$}}}=\mat I$ and therefore
$\mat{J_{U_j}}=\mat{D_j}$ has to be diagonal with entries $\pm 1$.

Consider
\begin{equation*}\label{eq:MPSInvolutionDiag}
\mat{A_j^{(i_j)}} \stackrel{\myref{eq:MPSRelationInvolution}}{=} \mat{U_j A_j^{(\bar{i}_j)} U_{j+1}}
\stackrel{\myref{eq:diagonalizeInvolution}}{=} \left( \mat{S_j^{\text{$-1$}} D_{j;\pm 1} S_j} \right)
            \mat{A_j^{(\bar{i}_j)}} \left( \mat{S_{j+1}^{\text{$-1$}} D_{j+1;\pm 1} S_{j+1}} \right) \; ,
\end{equation*}
which results in
\begin{equation}\label{eq:MPSInvolutionNormalForm}
\underbrace{\mat{S_j A_j^{(i_j)}} \mat S_{\mat{j+1}}^{-1}}_{=\mat{\tilde A_j^{(i_j)}}} =
\mat{D_{j;\pm 1}} \underbrace{\left( \mat{S_j A_j^{(\bar{i}_j)} S_{j+1}^{\text{$-1$}}}
                            \right)}_{= \mat{\tilde A_j^{(\bar{i}_j)} }} \mat{ D_{j+1;\pm 1}} \;
\end{equation}
showing that the \MPS matrices can be chosen such
that the involutions in Eqn.~\myref{eq:MPSRelationInvolution} can be expressed by diagonal matrices $\mat{D_{j;\pm 1}}$
yielding
$$ \mat{A_j^{(i_j)}} = \mat{D_{j;\pm 1} A_j^{(\bar{i_j})} D_{j+1;\pm 1}} \; .$$
Often the distribution of $\pm 1$ in $\mat{D_{j}}$ may be unknown.
So the exchange matrix
$\mat{J}=\mat{S}^{-1} \mat{D_J S}$ is an involution where as diagonal entries in $\mat{D_J}$, $+1$ and $-1$ appear
$\geq \lfloor {\operatorname{size}}(\mat{J})/2 \rfloor$.
If we double the allowed size $D$ for the \MPS matrices we can expect that
$\mat{J}$ has
at least as many $+1$ and $-1$ eigenvalues as all the appearing diagonal matrices $\mat{D_{j; \pm 1}}$.
Therefore,
we may heuristically replace each $\mat{D_{j; \pm 1}}$ by $\mat{J_j}$ with larger matrix size $D$
leading to an ansatz requiring no a-priori information.

\subsubsection*{Bit-Flip Symmetry in TI Representations}
If the \MPS matrices fulfill the bit-flip symmetry relations~\myref{eq:MPSRelationInvolution}
and are additionally site-independent, one has
$$ \mat{A^{(\bar i_j)}} = \mat{U} \mat{A^{(i_j)}} \mat{U}$$
with site-independent involutions
$\mat{U} \stackrel{\myref{eq:diagonalizeInvolution}}{=} \mat{S}^{-1} \mat{D_{\pm 1} S}$.
The transformation~\myref{eq:MPSInvolutionNormalForm} then reads
\begin{equation*}
\underbrace{\mat{S A^{(i_j)}} \mat{S}^{-1}}_{=\mat{\tilde A^{(i_j)}}} =
\mat{D_{\pm 1}} \underbrace{\left( \mat{S A^{(\bar{i}_j)} S^{\text{$-1$}}}
                            \right)}_{= \mat{\tilde A^{(\bar{i}_j)} }} \mat{ D_{\pm 1}} \; .
\end{equation*}
Thus, the vector can be represented by a \TIMPS fulfilling
$\mat{A^{(\bar i_j)}} = \mat{D_{\pm 1}} \mat{A^{(i_j)}} \mat{D_{\pm 1}}$
with the same involution $\mat{D_{\pm 1}}$ everywhere.
Similar results can be found in \cite{Sandvik07Variational}.

The $2D \times 2D$ involution~\myref{eq:BitFlipSymmetryProofConstructInvolution} from the proof
has as eigenvalues as many $+1$ as $-1$
and thus the related diagonal matrix $\mat{D_{\pm 1}}$ can be written as $\operatorname{diag}(\mat{I_D},-\mat{I_D})$.
Therefore, instead of $\mat{D_{\pm 1}}$ we may also use the $2D \times 2D$ exchange matrix~$\mat{J}$ as ansatz for an involution.

\subsubsection*{Uniqueness Results for the Bit-Flip Symmetry}
The technical remarks Lemmata \ref{lemmaTraceEqual}, \ref{LemmaCommute1} and \ref{LemmaCommute2}
may be put to good use in the following theorem.
It depicts certain necessary relations for the \MPS matrices to represent symmetric vectors.
\begin{theorem}
Let $p>1$.
Assume that the \MPS matrices (over $\mathbb K$) are related by
\begin{equation*}\label{eq:ConditionQuasiUniqueness}
\mat{A_1^{(1)}} = \mat{U_{p} A_1^{(0)} V_{1}}
\qquad \text{and} \qquad
\mat{A_j^{(1)}} = \mat{U_{j-1} A_j^{(0)} V_{j}} \ \text{ for } j=2,\dots,p
\end{equation*}
with square matrices $\mat{V_j}$ and $\mat{U_j}$ of appropriate size ($j=1,\dots,p$).
If any choice of matrices $\mat{A_j^{(0)}}$ results in the symmetry of the represented vector $\vec{x}$,
$$ \mat{J} \vec{x} = \vec{x}$$
then it holds $\mat{U_j}$ and $\mat{V_j}$ are -- up to a scalar factor -- involutions
for all $j$: $\mat{U}_{\mat{j}}^2=u_j\mat I$, $\mat{V}_{\mat{j}}^2=v_j\mat I$.
Furthermore, $\mat{U_{j}}=c_{j}\cdot \mat{V_j}$,
$j=1,...,p$ with constants $c_j$.
\end{theorem}

\begin{proof}
First, note that all $\mat{U_j}$ and $\mat{V_j}$ have to be nonsingular.
Otherwise, we could use a vector $\vec a \not= \vec 0$, e.g. with $\mat{U_{k-1}} \vec a = \vec 0$,
such that $\mat{A_k^{(0)}}=\vec a \vec b^{\herm}$ and
$\mat{A_k^{(1)}} = \mat{U_{k-1}A_k^{(0)}V_{k}}=\mat{0}$, giving $x_{1,1,...,1}=0$,
but with appropriate choice of the other $\mat{A_j^{(0)}}$ we can easily achieve $x_{0,0,...,0}\neq 0$.

Now, for all possible choices of $\mat{A_j^{(0)}}$, $j=1,...,p$, it holds
\begin{equation}\label{eq:x_0000eqx_1111}
\begin{split}
x_{1,1,...,1} & = \trace \left( \mat{A_1^{(1)}}\cdots \mat{A_p^{(1)}}\right)  =
\trace \left( \mat{U_pA_1^{(0)}V_1} \cdots \mat{U_{p-1} A_p^{(0)}V_p}\right) \quad \text{and}\\
x_{1,1,...,1} & = x_{0,0,...,0} = \trace \left( \mat{A_1^{(0)}} \cdots \mat{A_p^{(0)}}\right) \; .
\end{split}
\end{equation}
With the notation $\mat{W_j}=\mat{V_jU_{j}}$, $j=1,...,p$, 
we have
$$ \trace \left( \bigl(\mat{A_1^{(0)}}\cdots \mat{A_{p-1}^{(0)}}\bigr)\mat{A_p^{(0)}}\right) =
\trace \left( \bigl(\mat{W_pA_1^{(0)}W_1} \cdots \mat{A_{p-1}^{(0)}W_{p-1}}\bigr) \mat{A_p^{(0)}}\right)
$$
for all $\mat{A_p^{(0)}}$.
Therefore, Lemma \ref{lemmaTraceEqual} leads to
\begin{equation*}
\mat{A_1^{(0)}} \cdots \mat{A_{p-2}^{(0)}} \mat{A_{p-1}^{(0)}}
     = \mat{W_pA_1^{(0)}W_1} \cdots \mat{W_{p-2} A_{p-1}^{(0)}W_{p-1}}
\end{equation*}
and thus
\begin{align*}
\trace \left( \bigl( \mat{A_1^{(0)}}\cdots \mat{A_{p-2}^{(0)}} \bigr) \mat{A_{p-1}^{(0)}} \right) & =
\trace \left( \mat{W_pA_1^{(0)}W_1} \cdots \mat{W_{p-2} A_{p-1}^{(0)}W_{p-1}} \right) \\
 & = \trace \left( \bigl( \mat{W_{p-1}W_pA_1^{(0)}W_1} \cdots \mat{W_{p-2}} \bigr) \mat{A_{p-1}^{(0)}} \right) \; .
\end{align*}
If we proceed in the same way we iteratively reach
\begin{align}\label{eq:proofProceedIterativeMatrixequality}
\mat{A_1^{(0)}A_2^{(0)}}\cdots \mat{A_j^{(0)} } & =
    \mat{W_{j+1}} \cdots \mat{W_p} \mat{A_1^{(0)} W_1 A_2^{(0)}} \cdots \mat{W_{j-1}} \mat{A_j^{(0)} W_j}
                    \quad \text{and} \\ \label{eq:proofProceedIterativeTraceequality}
\trace \left( \mat{A_1^{(0)} A_2^{(0)}}\cdots \mat{A_j^{(0)}} \right) & =
    \trace \left(\mat{ W_jW_{j+1}} \cdots \mat{W_p} \mat{A_1^{(0)} W_1} \cdots \mat{W_{j-1}} \mat{A_j^{(0)}}\right) \; ,
\end{align}
for $j=p-1,...,1$.
Thus, we finally obtain the identities
\begin{align}\label{eq:proofFinallyObtainMatrixEquality}
\mat{A_1^{(0)}} & = \mat{W_2W_3} \cdots \mat{W_p} \mat{A_1^{(0)} W_1} \quad \text{and} \\ \label{eq:proofFinallyObtainTraceEquality}
\trace \left( \mat{A_1^{(0)}} \right) & = \trace \left( \mat{ W_1W_2} \cdots \mat{W_p} \mat{A_1^{(0)}} \right) \; ,
\end{align}
which hold for all $\mat{A_1^{(0)}}$. Lemma \ref{lemmaTraceEqual} applied to
\myref{eq:proofFinallyObtainTraceEquality} states
 $$\mat I = \mat{W_1W_2}\cdots \mat{W_p} \quad \text{or} \quad
 \mat{W_1^{\text{$-1$}}}=\mat{W_2}\cdots \mat{W_p} \; .$$
Inserting this in \myref{eq:proofFinallyObtainMatrixEquality} gives
$\mat{A_1^{(0)}} = \mat{W_1^{\text{$-1$}}A_1^{(0)}W_1}$
for all $\mat{A_1^{(0)}}$,
so due to Lemma \ref{LemmaCommute1}, $\mat{W_1} = w_1 \mat I$ for some constant $w_1 \not= 0$.
If we make use of this relation, Eqn. \ref{eq:proofProceedIterativeMatrixequality} (case $j=2$) leads to
$$ \mat{A_1^{(0)}A_2^{(0)}} =
    \left( w_1 \mat{W_{3}} \cdots \mat{W_p} \right) \left( \mat{A_1^{(0)}A_2^{(0)}} \right) \mat{W_2}
    \quad \text{for all } \mat{A_1^{(0)}}, \mat{A_2^{(0)}} \; .$$
Thus, Lemma \ref{LemmaCommute1} states $\mat{W_2} = w_2 \mat I$.
By induction,
Eqn. \ref{eq:proofProceedIterativeMatrixequality} reads
$$ \mat{A_1^{(0)}A_2^{(0)}}\cdots \mat{A_j^{(0)} } =
    \left( w_1 w_2 \cdots w_{j-1} \mat{W_{j+1}} \cdots \mat{W_p} \right) \left( \mat{A_1^{(0)}A_2^{(0)}}
        \cdots \mat{A_j^{(0)}} \right) \mat{W_j}$$
for all $\mat{A_1^{(0)}}, \cdots \mat{A_j^{(0)}}$, leading to
\begin{align}\label{eq:proofW_reqw_rI}
& \mat{W_j} = w_j \mat I \quad \text{or} \quad  \mat{U_{j}} =
w_j \mat{V_j^{\text{$-1$}}} \quad \text{ for all } \ j=1,\dots,p \; .
\end{align}
Considering instead $x_{1,0,...,0}=x_{0,1,...,1}$ gives
$$\trace \left( \mat{U_pA_1^{(0)}V_1A_2^{(0)}}\cdots \mat{A_p^{(0)}}\right) =
\trace \left( \mat{A_1^{(0)}U_1A_2^{(0)}V_2}\cdots \mat{U_{p-1}A_p^{(0)}V_p}\right) \; .$$
Replacing $\mat{A_1^{(0)}}$ by $\mat U_{\mat p}^{-1} \mat{A_1^{(0)}} \mat{V}_{\mat 1}^{-1}$
results in the above situation \myref{eq:x_0000eqx_1111}.
Analogously ($\mat{W_1} = \mat V_{\mat 1}^{-1} \mat{U_1}$, $\mat{W_p} = \mat{V_p} \mat U_{\mat p}^{-1}$)
one gets
$\mat{U_1}=c_1 \cdot \mat{V_1}$ and $\mat{U_p} = c_p \cdot \mat{V_p}$.

Repeating this technique at all positions $j$ for symmetries of the form
$ x_{i_1,\dots,i_{j-1},i_j,i_{j+1},\dots,i_p} =
    x_{0,\dots,0,1,0,\dots,0} = x_{1,\dots,1,0,1,\dots,1}$
gives the identities
\begin{alignat*}{3}
c_j \cdot \mat{V_j}  & = \mat{U_{j}} && \stackrel{\myref{eq:proofW_reqw_rI}}{=}
w_j \mat{V_j^{\text{$-1$}}} \quad \text{ for all } \ j=1,...,p \; .
\end{alignat*}
Therefore, all $\mat{U_j}$ and $\mat{V_j}$ are
involutions up to a factor,
$$ \mat{U_j^{\text{$2$}}} = w_j c_j \mat I \qquad \text{and} \qquad
\mat{V_j^{\text{$2$}}} = \frac{w_j}{c_j} \mat I \; .$$
Define $u_j := w_j c_j$ and $v_j:=\tfrac{w_j}{c_j}$ to finalize the proof.
\end{proof}
\begin{remark}
If we only allow unitary matrices $\mat{U_j}$ and $\mat{V_j}$
(e.g. $\mat{U_j}=\mat{V_j}=\mat J$ as motivated above),
the factors $c_j$ and $w_j$ (and thus also $u_j$ and $v_j$) have absolute value $1$.
\end{remark}

\subsubsection{Full-Bit Symmetry}
Now combine the previous symmetries and assume the following properties of the \MPS matrices
\begin{alignat}{3}\nonumber
\mat{A_j^{(0)}} & =  \mat{A} = \mat A^{\herm} && \quad \text{ for all} \ j \; \ \text{and} \\ \label{eq:FullBitPersymmetry}
\mat{A_j^{(1)}} & =  \mat{J A J} &&  \quad \text{ for all} \ j \; .
\end{alignat}
This ansatz results in reverse, bit-flip and bit-shift symmetry.

Neglecting the persymmetry~\myref{eq:FullBitPersymmetry} for the moment and only assuming
$$ \mat{A_j^{(0)}} = \mat{A^{(0)}} = \bigl( \mat{A^{(0)}} \bigr)^{ \herm} \quad \text{and} \quad \mat{A_j^{(1)}} =
\mat{A^{(1)}}= \bigl( \mat{A^{(1)}} \bigr)^{ \herm} \; ,$$
one may diagonalize $ \bigl( \mat{A^{(0)}} \bigr)^{ \herm} = \mat{A^{(0)}} = \mat{U}^{\herm} \mat{\Lambda U}$
and set $\mat{B} = \mat{U A^{(1)} U}^{\herm}$.
Hence, we propose to define a normal form of the type
$$ \mat{\tilde A_j^{(0)}} = \mat{\tilde A^{(0)}} = \mat{\Lambda} \qquad \text{and} \qquad
\mat{\tilde A_j^{(1)}} = \mat{\tilde A^{(1)}} = \mat{B} = \mat B^{\herm} \; .$$

\subsubsection{Reduction in the Degrees of Freedom}
The symmetries discussed in the previous paragraphs lead to a reduction of the number of free parameters.
First let us discuss the reduction in the
number of entries in the full vector $\vec{x}$.
The bit-shift symmetry $x_{i_1,i_2,\dots,i_p} = x_{i_2,\dots,i_p,i_1}=...$
reduces the number of different entries approximately to $p^{-1}\,2^p$.
Both bit-flip and reverse symmetry lead to a reduction factor $1/2$ in each case.
Note that not all of these symmetries are independent, e.g.,
the symmetry $x_{i_1,i_2} = x_{i_2,i_1}$ is a consequence of either the bit-shift or the reverse symmetry.
On the other hand the three symmetries are indeed independent in general.
To see this we consider the following example with $p=9$ binary digits:
$$ (i_1,i_2,i_3,i_4,i_5,i_6,i_7,i_8,i_9) =  (101001000) \; .$$
Table \ref{tab:symmetries} lists for all of the three classes of symmetries all index sets
which are related to equal vector components.
\begin{table}
\tbl{Listing index sets related to equal vector components for different symmetries.
This table shows that the bit-shift symmetry,
the bit-flip symmetry and the reverse symmetry are principally independent.}
{\begin{tabular}{@{}ll}\toprule
Bit-shift symmetry & $101001000$, $010010001$, $100100010$, $001000101$, $010001010$ \\
 & $100010100$,$000101001$, $001010010$, $010100100$ \\ \colrule
Bit-flip symmetry & $101001000$, $010110111$ \\ \colrule
Reverse symmetry & $101001000$, $000100101$\\
\botrule
\end{tabular}}
\label{tab:symmetries}
\end{table}

In the \MPS ansatz we have similar reductions.
The bit-shift symmetry uses one matrix pair instead of $p$, giving a reduction factor $p$.
The bit-flip symmetry has a reduction factor $2$ (if we ignore different choices
for $\mat{D_{j;\pm 1}}$), and in the reverse symmetry only half of the matrices
can be chosen.
Note, that this will not only lead to savings in memory but also to faster convergence
and better approximation in the applied eigenvalue methods because the representation
of the vectors has less degrees of freedom and allows a better approximation of the
manifold that contains the eigenvector we are looking for.

\subsubsection{Further Symmetries}
In this paragraph we analyze further symmetries such as

\begin{subequations}
\begin{minipage}{0.25\textwidth}
\begin{equation}\label{eq:FurtherSymmetryRelation1}
\vec{x} = \pmat{
                 \vec{b} \\
                 \vec{b} \\
               } \; ,
\end{equation}
\end{minipage}
\hfil
\begin{minipage}{0.25\textwidth}
\begin{equation}\label{eq:FurtherSymmetryRelation2}
\vec{x} = \pmat{
                \vec{b} \\
                \vec{-b} \\
              } \; ,
\end{equation}
\end{minipage}
\hfil
\begin{minipage}{0.3\textwidth}
\begin{equation}\label{eq:FurtherSymmetryRelation3}
\vec{x} = \pmat{                          b_1 \\
                                                      \pm b_1 \\
                                                      b_2 \\
                                                      \pm b_2 \\
                                                      \vdots \\
                                                    } \; .
\end{equation}
\end{minipage}
\end{subequations}\\
The following lemma states results for the symmetry~\myref{eq:FurtherSymmetryRelation1}.
\begin{lemma}
If the first matrix pair is of the type
\begin{equation}\label{eq:FurtherSymmetryRelationSite1}
\bigl( \mat{A_1^{(0)}} \ , \ \mat{A_1^{(1)}} \bigr)
      = \bigl( \mat{B} \ ,\ \mat{B} \bigr) \; ,
\end{equation}
the represented vector takes the form~\eqref{eq:FurtherSymmetryRelation1}
and, vice versa, any vector of the form~\eqref{eq:FurtherSymmetryRelation1}
can be expressed by an \MPS fulfilling \myref{eq:FurtherSymmetryRelationSite1}.
\end{lemma}
\begin{proof}
The given \mps relation~\myref{eq:FurtherSymmetryRelationSite1} implies
$x_{0,i_2,i_3,\dots,i_p} = x_{1,i_2,i_3,\dots,i_p}$ for all $i_2,\dots,i_p$.
Hence, the represented vector $\vec x$ is of the form~\eqref{eq:FurtherSymmetryRelation1}.

In order to specify an \MPS representation for a vector $\vec x$
fulfilling \eqref{eq:FurtherSymmetryRelation1}
we consider any \MPS representation (see, e.g., Lemma~\ref{lemma:ExistenceUniqueMPSGauge})
for the vector~$\vec b$,
$$
\vec b = \sum_{i_2,\dots,i_p} \trace \left( \mat{B_2^{(i_2)} B_3^{(i_3)}} \cdots \mat{B_p^{(i_p)}} \right)
            \vec{e_{i_2,i_3,\dots,i_p}} \; .
$$
The definition $\mat{B_1^{(0)}} = \mat{B_1^{(1)}} = \mat{I_{D_2}}$ results in the desired relations
$$ x_{0,i_2,\dots,i_p} = x_{1,i_2,\dots,i_p} =
\trace \left( \mat{B_1^{(i_1)} B_2^{(i_2)}} \cdots \mat{B_p^{(i_p)}} \right) =
\trace \left( \mat{B_2^{(i_2)}} \cdots \mat{B_p^{(i_p)}}\right) = b_{i_2,\dots,i_p} \; .$$
\end{proof}
\begin{remark}
\begin{enumerate}
\item The proof works for PBC and OBC.
In the latter case $\mat{B_1^{(i_1)}}$ specializes to a scalar.
\item The second symmetry~\myref{eq:FurtherSymmetryRelation2} corresponds to the relation
$\mat{A_1^{(1)}} = - \mat{A_1^{(0)}}$.
Adapting the proof to this case would give $\mat{B_1^{(0)}}=\mat{I_{D_2}}$ and $\mat{B_1^{(1)}}=-\mat{I_{D_2}}$.
\item The symmetry type~\myref{eq:FurtherSymmetryRelation3} is related to
$\mat{A_p^{(1)}} = \pm \mat{A_p^{(0)}}$.
The construction would analogously read
$\mat{B_p^{(0)}}=\mat{I_{D_p}}$ and $\mat{B_p^{(1)}}=\pm \mat{I_{D_p}}$.
\item Similarly, we can impose conditions on the \MPS representation that certain
local matrix products are equal resulting in symmetry properties of $\vec{x}$.
So the condition $\mat{A_1^{(0)}A_2^{(0)}}=\mat{A_1^{(1)}A_2^{(1)}}$ leads to
$x_{0,0,i_3,...,i_p}\equiv x_{1,1,i_3,...,i_p}$.\\
Imposing the conditions
$\mat{A_1^{(0)}A_2^{(0)}}=\mat{A_1^{(1)}A_2^{(1)}}=\mat{A_1^{(0)}A_2^{(1)}}=\mat{A_1^{(1)}A_2^{(0)}}$
leads to the symmetry
$x_{0,0,i_3,...,i_p}\equiv x_{1,1,i_3,...,i_p}\equiv x_{0,1,i_3,...,i_p}\equiv x_{1,0,i_3,...,i_p}$.
\end{enumerate}
\end{remark}
In the following theorem we state certain necessary relations for the \MPS representation of symmetries,
which are of the form~\myref{eq:FurtherSymmetryRelation1}.
\begin{theorem}\label{theorem:Symmetryb_b}
Assume that the \MPS matrices (over $\mathbb K$) are related via
\begin{equation*}\label{eq:ConditionQuasiUniqueness2}
\mat{A_1^{(1)}} = \mat{VA_1^{(0)} U}
\end{equation*}
with matrices $V$ and $U$.
If any choice of matrices $\mat{A_j^{(0)}}$, $j=1,...,p$ for fixed  $\mat{A_j^{(1)}}$, $j>1$,
results in a vector $\vec{x}$ of the form
$$ \vec{x} = \left(
              \begin{array}{c}
                \vec{b} \\
                \vec{b} \\
              \end{array}
            \right) $$
then $\mat{U} =c\mat{I}= \mat{V}^{-1}$ and so $\mat{A_1^{(1)}}=\mat{A_1^{(0)}}$.
\end{theorem}

\begin{proof}
The assumption leads to the equation
$$
\trace \left( \bigl( \mat{A_1^{(0)}}-\mat{VA_1^{(0)}U} \bigr) \mat{A_2^{(i_2)}}\cdots \mat{A_p^{(i_p)}} \right) \equiv 0
$$
for all choices of matrices $\mat{A_j^{(i_j)}}$, $j>2$.
From Lemma \ref{lemmaTraceEqual} we obtain
$\mat{A_1^{(0)}}=\mat{VA_1^{(0)}U}$ for all choices of $\mat{A_1^{(0)}}$.
Hence, due to Lemma \ref{LemmaCommute1},
$\mat{U} = c \mat{I}$ and $\mat{V} = \mat{I}/c$ for a nonzero $c$.
This gives $\mat{U} = c \mat{I} = \mat{V}^{-1}$.
\end{proof}
\begin{remark}
The result of Theorem \ref{theorem:Symmetryb_b} can be easily adapted to the case \myref{eq:FurtherSymmetryRelation2}.
Moreover, it can be generalized to symmetries such as \myref{eq:FurtherSymmetryRelation3},
which are of the form
$x_{i_1,\dots,i_{r},0,i_{r+2},\dots,i_p} = x_{i_1,\dots,i_{r},1,i_{r+2},\dots,i_p}$.
\end{remark}

\subsubsection{Closing remarks on symmetries}
Let us conclude this paragraph on symmetries with some remarks on applications.
So far, we have seen that there are different symmetries which can be represented by
convenient relations between the \mps matrices.
Furthermore we proposed convenient normal forms and attested related uniqueness results.
It is more difficult to exploit such symmetry approaches in numerical algorithms
such as eigenvector approximation.
As standard methods like DMRG (\cite{SchollDMRG2011}) usually do not preserve our proposed symmetries,
one would have to consider other techniques such as gradient methods,
which are already in use in QI groups (see, e.g., \cite{Pirvu11ExploitingTI}).

\section{Conclusions}
\label{sec:conclusions}
Based on a summary of definitions and properties of structured matrices,
we have analyzed matrix symmetries as well as symmetries induced by open or periodic
boundary conditions (as well as their interdependence).
In order to describe symmetry relations in physical 1D many-body quantum systems
by Matrix Product States or Tensor Trains we have developed efficient representations.
To this end, normal forms of \mps in a general setting as well as in special symmetry relations
have been introduced that may be useful to cut the number of degrees of freedom of $p$ two-level systems down
and may lead to better theoretical representations as well as more efficient numerical algorithms.

\section*{Acknowledgements}
\label{sec:acknowledgements}

This work was supported in part by the Bavarian excellence network ENB
via the International Doctorate Program of Excellence
{\em Quantum Computing, Control, and Communication} (QCCC).



\end{document}